\documentclass{fundam}

\usepackage{amsmath}
\usepackage{amssymb}
\usepackage{geometry}
\usepackage{fancyhdr}
\usepackage{amsfonts}
\usepackage{times}
\usepackage{subcaption}
\usepackage{todonotes}

\usepackage{tikz}
\usetikzlibrary{arrows,shapes,snakes,automata,backgrounds,petri,positioning}
\tikzset{>=latex}

\usepackage{graphicx}
\usepackage{algorithmicx}
\usepackage{algpseudocode}

\usepackage{scalerel}
\newcommand{\N}{\scaleto{N}{3.2pt}}
\newcommand{\C}{\scaleto{C}{3.2pt}}
\newcommand{\U}{\scaleto{U}{3.2pt}}
\renewcommand{\S}{\scaleto{S}{3.2pt}}

\newcommand{\1}{\scaleto{1}{3.2pt}}

\renewcommand{\phi}{\varphi}
\newcommand\Tau{\mathcal{T}}

\newcommand{\gdw}{\Leftrightarrow}
\newcommand\abs[1]{\vert #1 \vert}
\newcommand{\pre}[1]{\bullet #1}
\newcommand{\post}[1]{#1 \bullet}

\newcommand{\scopen}[1]{#1 \in \{1,\dots,n\}}

\begin{document}

\setcounter{page}{245}
\publyear{22}
\papernumber{2138}
\volume{187}
\issue{2-4}

 \finalVersionForARXIV

\title{Skeleton Abstraction for Universal Temporal Properties}

\author{Sophie Wallner\thanks{Address for correspondence: University of Rostock, 18051 Rostock, Germany},  Karsten Wolf \\
University of Rostock, Germany\\
sophie.wallner@uni-rostock.de
}

\runninghead{S. Wallner and K. Wolf}{Skeleton Abstraction for Universal Temporal Properties}

\maketitle

\begin{abstract}
		Uniform coloured Petri nets can be abstracted to their \emph{skeleton}, the place/transition net that simply turns the coloured tokens into black tokens.
		A coloured net and its skeleton are related by a \emph{net morphism} \cite{desel91,padberg98}.
		For the application of the skeleton as an abstraction method in the model checking process, we need to establish a \emph{simulation relation} \cite{milner89}
		between the state spaces of the two nets.
		Then, universal temporal properties (properties of the $ ACTL^* $ logic) are preserved.
		The abstraction relation induced by a net morphism is not necessarily a simulation relation, due to a subtle issue related to deadlocks \cite{findlow92}.
		We discuss several situations where the abstraction relation induced by a net morphism is as well a simulation relation, thus preserving $ACTL^*$ properties.
		We further propose a partition refinement algorithm for folding a place/transition net into a coloured net. This way, skeleton abstraction becomes available for models given as place/transition nets.
		Experiments demonstrate the capabilities of the proposed technology. Using skeleton abstraction, we are capable of solving problems that have not been solved before in the Model Checking Contest \cite{mcc19}.
\end{abstract}

\section{Introduction} \label{sec:intro}

In the model checking process for coloured Petri nets, one of the biggest issues is the state explosion problem, which makes the verification of a property impossible, as the state space is getting too big to handle.
A way to deal with these big systems, is the well known technique of abstraction.
Given a coloured Petri net $ C $, we can form its \emph{skeleton} $ S $, which has the structure of $ C $ and simply decolours its components and tokens.
This skeleton is an abstraction of the coloured net.
To use this abstraction technique in the model checking process, we need to guarantee, that properties are preserved through this abstraction, i.e. that the validity of a property in $ S $ indicates the validity of the property in $ C $.
Unfortunately, this is not the case for every coloured net.
The issue is that some deadlocks of $ C $ are not preserved in $ S $, as the additional behaviour of $ S $ changes the validity of the property.
Deadlocks in a coloured net can have two different causes.
First, they can be caused by an insufficient number of tokens in the preset of a transition.
These deadlocks are preserved in the skeleton, as the number of tokens will neither be sufficient in the skeleton.
Second, they can be caused by a wrong colour set of tokens, as the number of tokens in the preset of a transition is sufficient, but the colour distribution of the tokens violates the guard of the transition.
This type of deadlocks is usually not preserved in the skeleton, as the skeleton does not distinguish colors at all.
Consider the following example:

\begin{example}
	\label{ex:1}
	Let $ C $ be a coloured Petri net, for which we build its skeleton $ S $ by removing the colour sets of the places, the guard of the transition and making the tokens all indistinguishable.
	The two nets are pictured in Figure~\ref{fig:blocked}.
	We consider the $ ACTL^* $ formula $ \phi: \mathbf{A} \mathbf{F} \, p \leq 1 $.
	The guard of the only transition $ t $ expresses that $ t $ requires three tokens of the same colour to be activated and then produces one token of this colour.
	In the given marking of $ C $, $ t $ is not enabled, so this marking is a deadlock.
	The corresponding marking of $ S $ is not a deadlock, as the number of tokens is sufficient and $ t $ is activated.
	Firing $ t $ in $ S $ leads to a marking, where all tokens are removed from $ p $, so $ \phi $ is true for $ S $.
	However, $ \phi $ is not true for $ C $.
	Transferring the validity of $ \phi $ from $ S $ to $ C $ will draw a wrong conclusion.

\begin{figure}[h]
	\begin{subfigure}{0.5\textwidth}
		\begin{center}
			\begin{tikzpicture}[auto,token/.style={shape=circle,draw=black,fill=white,minimum size=2.5mm,inner sep=0.5pt}]
			\node[place,label=above:$p$]	(p1) {};
			\node[token] at (0,0.21) {\tiny r};
			\node[token] at(-0.16,-0.1) {\tiny g};
			\node[token] at(0.16,-0.1) {\tiny g};
			\node[place,label=above:$q$] at(5,0)	(p2){};
			\node[transition, label=above:$t$] at(2.5,0) (t){}
				edge[pre] node[swap]{$ x_1 \: x_2 \: x_3 $} (p1)
				edge[post] node{$ y $} (p2);

			\node at (0,-0.7) (g) {$\scriptstyle{\chi(p) = \{r,g\}}$};
			\node at (5,-0.7) (g) {$\scriptstyle{\chi(q) = \{r,g\}}$};

			\node at (2.5,-0.5) (g) {$\scriptstyle{\gamma(t) = \bigvee_{c \in \{r,g\}}}$};
			\node at (2.5,-0.9) (g) {$\scriptstyle{x_1 = x_2 = x_3 = c}$};
			\node at (2.5,-1.3) (g) {$\scriptstyle{ \land \, y = c}$};
			\end{tikzpicture}
		\subcaption{A coloured net $ C $.}
		\end{center}
	\end{subfigure}
	\begin{subfigure}{0.5\textwidth}
		\begin{center}
			\begin{tikzpicture}[auto]
			\node[place,label=above:$p$,tokens=3]	(p1) {};
			\node[place,label=above:$q$] at(5,0)	(p2){};
			\node[transition, label=above:$t$] at(2.5,0) (t){}
				edge[pre] node[swap]{$ 3 $} (p1)
				edge[post] node{$ 1 $} (p2);
			\node at (2,-0-.5) (g) {\textcolor{white}{$\scriptscriptstyle{\bigvee_{i \in \{g,r\}} x_1 = x_2 = x_3 = i \land y = i}$}};
		\end{tikzpicture}
		\subcaption{The skeleton $ S $ of $ C $.}
		\end{center}
	\end{subfigure}\vspace*{-1mm}
	\caption{A coloured net $ C $ with a deadlock not preserved in its skeleton $ S $.}
	\label{fig:blocked}\vspace*{-2mm}
\end{figure}
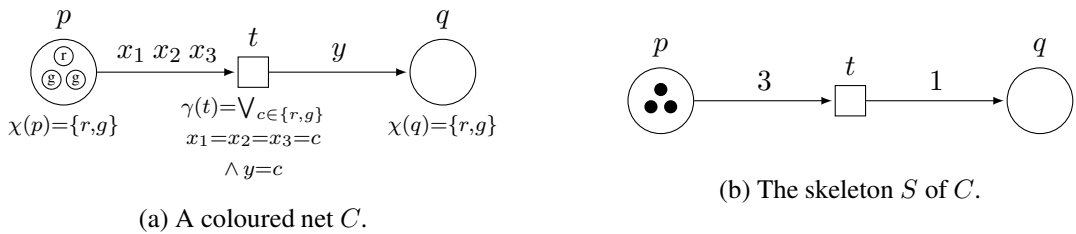
\end{example}

This paper is an improved version of the conference paper \cite{ourskeleton} and gives a detailed analysis of situations, where the skeleton as an abstraction technique is soundly applicable in the model checking process for coloured Petri nets.
With respect to \cite{ourskeleton}, we improved the folding algorithm in Section~\ref{subsec:fold} in a way, that it is now more liberal, i.e. the resulting coloured nets are smaller.
The efficiency of this new folding algorithm is examined through new experiments.
Overall, the more liberal folding algorithm helps to generate more results.
We as well formalized the folding algorithm in Section~\ref{subsec:fold} and added some explanation for this algorithm to make the construction more understandable.
Furthermore, we extended the description of the automaton approach to check whether a transition class is full or not.
We therefore introduced more formal definitions of the arc inscriptions and guards.
This topic is now covered in an extra section (Section~\ref{sec:fullnesscheck}) in which we describe how to generate an automaton for terms resp. guard expressions and finally a transition class.
This will make the process of checking fullness more understandable.
Based on our new folding algorithm, Section~\ref{sec:experiments} describes updated experimental results.
The paper is structured as follows:
The next Section~\ref{sec:related} will give an overview of the application of skeleton nets in other contexts, Section~\ref{sec:basic} provides necessary basic definitions.
After that, Section~\ref{sec:relations} will introduce the different concepts of relations, which can hold between reachability graphs of Petri nets.
Section~\ref{sec:sim} will set the focus on the simulation relation between reachability graphs as the core concept for keeping validity through abstraction.
We present a survey, in which cases the skeleton is a valid abstraction method for a coloured Petri net, distinguishing different classes of nets and types of formulas.
Section~\ref{sec:fullnesscheck} provides an approach to check whether a coloured net has a deadlock-preserving skeleton; a property a net might fulfill to make the skeleton abstraction a valid verification extension.
With the folding algorithm in Section~\ref{sec:ext}, we extend the scope of application of the skeleton abstraction to place/transition nets.
The experimental results in Section~\ref{sec:experiments} underline the powerfullness of this abstraction method.

\section{Related work} \label{sec:related}

The idea of a skeleton-based analysis of a Petri net is subject of \cite{vautherin87}.
Based on this, \cite{findlow92} examines the role of deadlocks within this topic more precisely.
The results are also applied in other contexts.
In \cite{rusttackenboeke01} extended Pr/T-Nets are used as a modeling formalism for embedded real-time systems due to the multitude of analysis methods for Petri nets. With the skeleton of a Pr/T-Net, properties like reachability of states or deadlock freeness can be examined.
\cite{lilius95} transfers Findlow's results \cite{findlow92} to algebraic nets. They are used as an application for a folding construction, which is described there.
Findlow's observations on deadlock-preserving skeletons are further used in \cite{silva99} for a skeleton-based analysis of G-Nets, an object-based Petri net formalism.
The preservation of predicates in temporal logic under morphisms has also already been discussed.
\cite{padberg98} describes a rule-based modification of algebraic high-level nets extended with morphisms such that safety properties described in temporal logic are preserved.This provides a technique which allows to transfer safety properties between the source and the target net.

\section{Basic definitions} \label{sec:basic}

First, we present definitions for place/transition nets.

\begin{definition}[Place/Transition Net]
	A \emph{place/transition net} (P/T net) is a tuple $ N = [P,T,F,W,m_0] $, where $ P $ is a finite set of \emph{places} and $ T $ is a finite set of \emph{transitions} with $ P \cap T = \emptyset $.
	The \emph{arcs} $ F \subseteq (P \times T) \cup (T \times P) $ of the net are labeled by a \emph{weight function} $ W: F \rightarrow \mathbb{N} $, with $ W(x,y) = 0 $ iff $ (x,y) \notin F $.
	A marking is a mapping $ m : P \rightarrow \mathbb{N} $ and $ m_0 $ is the \emph{initial marking}.
\end{definition}

The behavior of a P/T net is defined by the transition rule.

\begin{definition}[Transition rule of a P/T net]
	Let $ N = [P,T,F,W,m_0] $ be a P/T net.
	Transition $ t \in T $ is enabled in marking $ m $ if $ \forall p \in P: W(p,t) \leq m(p)$.
	Firing Transition $ t $ leads from marking $ m $ to marking $ m' $ (denoted as $ m \xrightarrow{t} m' $) in $ N $ , if $ t $ is enabled in $ m $ and
	 $ \forall p: m'(p) = m(p) - W(p,t) + W(t,p) $.
\end{definition}

A marking $ m' $ is \emph{reachable} from a marking $ m $ (denoted as $ m \xrightarrow{*} m'$), if there is a firing sequence $ t_1t_2\dots t_n \in T^*$, such that $ m \xrightarrow{t_1} m_1 \xrightarrow{t_2} \dots \xrightarrow{t_n} m'$.
We extend the notation of reachability to firing sequences $\omega \in T^*$ and we call $RS(N)=\{m \mid \exists \omega \in T^*: m_0\xrightarrow{\omega}m\}$ the \emph{reachability set}, which contains all of $N$'s reachable markings.
Using the transition rule, a Petri net induces a labeled transition system, called the \emph{reachability graph}.

\begin{definition}[Labeled Transition System, Reachability Graph]
	A \emph{transition system} $ TS = [Q,q_o,R,A] $ is a labeled, directed graph, where $ Q $ is the set of \emph{states}, $ q_0 \in Q $ is the initial state and a transition relation $ R \subseteq Q \times A \times Q $ with some set of actions $ A $.
	The \emph{reachability graph} $ R_N(m_0) $ of a Petri net $ N $ is a transition system, where the set of states is $ RS(N) $, $ m_0 $ serves as the initial state and $ (m,t,m') \in R $ iff $ m \xrightarrow{t} m'$.
\end{definition}

Furthermore, we introduce a simple notion for coloured Petri nets with finite colour domains.

\begin{definition}[Coloured Petri net]
	A \emph{coloured Petri net} $ C = [P_{\C},T_{\C},F_{\C},W_{\C},\chi,\gamma,m_{0\C}] $ consists of a finite set $ P_{\C} $ of \emph{places}, a finite set $ T_{\C} $ of \emph{transitions} where $ P_{\C} \cap T_{\C} = \emptyset$ and a set of \emph{arcs} $ F_{\C} \subseteq (P_{\C} \times T_{\C}) \cup (T_{\C} \times P_{\C})$.
	The \emph{weight function} $ W_{\C} $ assigns a finite set of variables to each element of $ F_{\C} $.
	If $ (x,y) \notin F_{\C} $, we assume $ W_{\C}(x,y) = \emptyset $.
	The \emph{colouring function} $ \chi $ assigns a finite set $ \chi(p) $ of colours to each place $ p \in P_{\C} $, called colour domain of $ p $.
	The \emph{guard function} $ \gamma $ assigns a boolean predicate $ \gamma(t) $ to each transition $ t \in T_{\C} $, which ranges over the variables of $ W_{\C}(p,t) \cup W_{\C}(t,p) $ for all $ p \in P_{\C} $.
	The initial marking $ m_0 $ is a multiset over $ \chi(p) $ for every $ p \in P_{\C}$.
	The number of tokens of colour $ c $ on place $ p $ in marking $ m $ is described as $ m(p)(c) $.
\end{definition}

For a transition $ t \in T_{\C}$, we define a \emph{firing mode} of $ t $ as a mapping $ g: \bigcup_{p \in P_{\C}} (W_{\C}(p,t) \cup W_{\C}(t,p)) \rightarrow \bigcup_{p \in P_{\C}} \chi(p) $, which assigns a colour from $ \chi(p) $ for every place $ p \in P_{\C}$ and for each variable $ x \in W_{\C}(p,t) \cup W_{\C}(t,p) $.
A firing mode $ g $ of a transition $ t $ satisfies the guard $ \gamma(t) $, denoted as $ g \models \gamma(t) $, if the assignment of colours to variables is a model of the guard.

Usually, definitions of coloured nets permit a richer syntax for arc weights and provide a more detailed description of the guard.
In Chapter \ref{sec:fullnesscheck} we give more specific definitions for arc weights and guards; furthermore we present a solution how to simplify arc weights to variables without undermining expressivity.
Until then, consider arc weights and guards as defined in the definition above.

For a coloured net, we define its unfolding.

\begin{definition}[Unfolding]
	Let $ C = [P_{\C},T_{\C},F_{\C},W_{\C},\chi,\gamma,m_{0\C}] $ be a coloured Petri net.
	A P/T net $ U = [P_{\U},T_{\U},F_{\U},W_{\U},m_{0\U}] $ is the \emph{unfolding} of $ C $ if
	\begin{itemize}
\itemsep=0.9pt
		\item $ P_{\U} = \{ [p,c] \mid p \in P_{\C}, c \in \chi(p)\} $
		\item $ T_{\U} =  \{ [t,g] \mid t \in T_{\C}, g \models \gamma(t)\} $
		\item $ ([p,c], [t,g]) \in F_{\U}$, iff $ (p,t) \in F_{\C}$ and $c \in g(W_{\C}(p,t)) $
		\item $ ([t,g], [p,c]) \in F_{\U}$, iff $ (t,p) \in F_{\C}$ and $c \in g(W_{\C}(t,p)) $
		\item $ W_{\U}([p,c],[t,g]) = \operatorname{card}(\{x \mid x \in W_{\C}(p,t), g(x) = c\}) $
		\item $ W_{\U}([t,g],[p,c]) = \operatorname{card}(\{x \mid x \in W_{\C}(t,p), g(x) = c\}) $
		\item $ m_{0\U}([p,c]) = m_{0\C}(p)(c)$.
	\end{itemize}
\end{definition}

In the sequel, refer to the transition system defined by a coloured net $ C $ as the transition system of its unfolding $ U $, as they are isomorphic \cite{jensen09}.
Coloured nets as defined above are \emph{uniform}.
This means that the number of tokens consumed or produced by a transition is independent of the particular firing mode, i.e. always $\operatorname{card}(W(x,y))$ tokens.
There exist non-uniform variants of coloured nets.
They use variables that take multisets over $ \chi(p) $ as values.
They are, however, out of the scope of this article since the core artifact studied in this paper, the skeleton, is not applicable to non-uniform nets.
For a uniform net, we can assign a second P/T net, its skeleton.

\begin{definition}[Skeleton]
	Let $ C = [P_{\C},T_{\C},F_{\C},W_{\C},\chi,\gamma,m_{0\C}]$ be a coloured net.
	Its \emph{skeleton} $ S = [P_{\S},T_{\S},F_{\S},W_{\S},m_{0\S}]$ is a P/T net where
	\begin{itemize}
\itemsep=0.95pt
	\item $P_{\S} = P_{\C}$, $T_{\S} = T_{\C}$, $F_{\S} = F_{\C}$
	\item for all $x,y \in P \cup T : W_{\S}(x,y) = \operatorname{card}(W_{\C}(x,y)) $
	\item for all $p \in P: m_{0\S}(p) = \sum_{c \in \chi(p)} m_0(p)(c)$.
	\end{itemize}
\end{definition}

The following example will help to understand the concepts of the unfolding and the skeleton of a coloured net.
\begin{sloppypar}
	\begin{example}
	Let $ C $ be the given coloured Petri net, as depicted in Figure~\ref{fig:faltungentfaltung}.
	Place $ p $ and $ q $ have the colour domain $ \chi(p) = \chi(q) =\{r,g,b\}$.
	Unfolding $ C $ leads to the corresponding places $ [p,g],[p,r],[p,b] $ resp. $ [q,g],[q,r],[q,b] $.
	For every firing mode, which satisfies the guard of transition $ t $, we introduce one transition in the unfolding, so the unfolding has three transitions $ t_g,t_r,t_b$.
	Building the skeleton makes all tokens on $ p $ indistinguishable and removes the colour sets of $ p $ and $ q $.
	The transition $ t $ in the skeleton has no guard and is simply activated, if there is a sufficient number of tokens on $ p $.
	
\begin{figure}[h]
		\resizebox{\textwidth}{!}{
			\begin{subfigure}{0.42\textwidth}
				\begin{center}
					\begin{tikzpicture}
					\node[place,tokens=1]	at(0,0) (pg1) {};
						\node at (0,0.7) (lpg1) {$ [p,g] $};
					\node[place,tokens=1]	at(1,0) (pr1) {};
						\node at (1,0.7) (lpg1) {$ [p,r] $};
					\node[place,tokens=1]	at(2,0) (pb1) {};
						\node at (2,0.7) (lpg1) {$ [p,b] $};

					\node[place]	at(0,-3) (pg2) {};
						\node at (0,-3.7) (lpg1) {$ [q,g] $};
					\node[place]	at(1,-3) (pr2) {};
						\node at (1,-3.7) (lpg1) {$ [q,r] $};
					\node[place]	at(2,-3) (pb2) {};
						\node at (2,-3.7) (lpg1) {$ [q,b] $};

					\node[transition] at (0,-1.5) (tg) {$ t_g $}
						edge[pre] (pg1)
						edge[post] (pg2);
					\node[transition] at (1,-1.5) (tr) {$ t_r $}
						edge[pre] (pr1)
						edge[post] (pr2);
					\node[transition] at (2,-1.5) (tb) {$ t_b $}
						edge[pre] (pb1)
						edge[post] (pb2);
				\end{tikzpicture}
				\subcaption{The Unfolding $ U $.}
				\end{center}
			\end{subfigure}
			\begin{subfigure}{0.3\textwidth}
				\begin{center}
					\begin{tikzpicture}[wtoken/.style={shape=circle,draw=black,fill=white,minimum size=2.5mm,inner sep=0.5pt}]
					\node[place]	at(0,0) (pg1) {};
						\node at (0,0.7) (lpg1) {$ p $};
						\node[wtoken] at (0,0.21) {\tiny r };
						\node[wtoken] at (0.16,-0.1) {\tiny b };
						\node[wtoken] at (-0.16,-0.1) {\tiny g };
						\node at (1.5,0) {$ \scriptstyle \chi(p) = \{r,g,b\}$};

					\node[place]	at(0,-3) (pg2) {};
						\node at (0,-3.7) (lpg1) {$ q $};
						\node at (1.5,-3) {$ \scriptstyle \chi(q) = \{r,g,b\}$};

					\node[transition] at (0,-1.5) (tg) {$ t $}
						edge[pre] (pg1)
						edge[post] (pg2);
						\node at (2.3,-1.5) {$ \scriptstyle \gamma(t) = \bigvee_{c \in \{r,g,b\}} x_1 = c \, \land \, x_2 = c$};
						\node at(0.3,-0.8) (xpt) {$\scriptstyle x_1 $};
						\node at(0.3,-2.1) (xtq) {$\scriptstyle x_2 $};
				\end{tikzpicture}
				\subcaption{The Coloured Net $ C$.}
				\end{center}
			\end{subfigure}
			\begin{subfigure}{0.3\textwidth}
				\begin{center}
					\begin{tikzpicture}
					\node[place,tokens=3]	at(0,0) (pg1) {};
						\node at (0,0.7) (lpg1) {$ p $};

					\node[place]	at(0,-3) (pg2) {};
						\node at (0,-3.7) (lpg1) {$ q $};

					\node[transition] at (0,-1.5) (tg) {$ t $}
						edge[pre] (pg1)
						edge[post] (pg2);
				\end{tikzpicture}
				\subcaption{The Skeleton $ S $.}
				\end{center}
			\end{subfigure}
	}
	\caption{A coloured Petri net $ C $, its unfolding $ U $ and its skeleton net $ S $.}
	\label{fig:faltungentfaltung}
\end{figure}
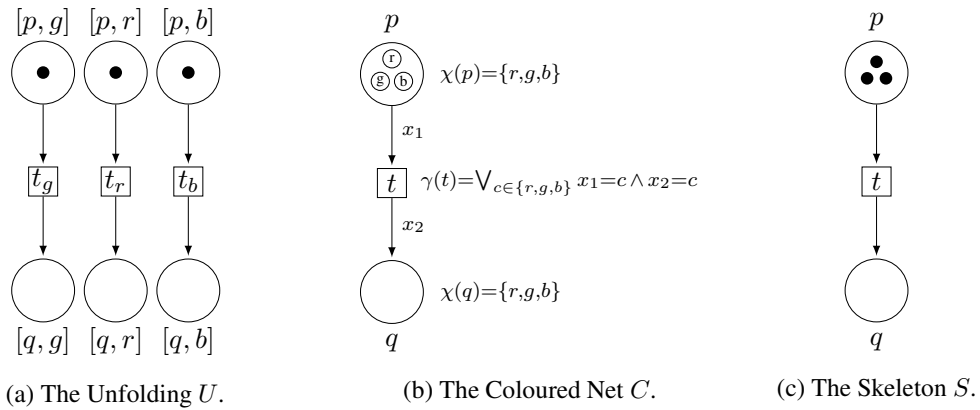

	\end{example}
\end{sloppypar}

In the sequel, unless stated otherwise, let $ C $ be an arbitrary but fixed coloured net, $ U $ its unfolding, and $ S $ its skeleton.
$ U $ and $ S $ are related by a net morphism.

\begin{definition}[Net Morphism \cite{desel91}] \label{def:morphism}
Let $N_1 = [P_1,T_1,F_1,W_1,m_{01}]$ and $N_2= [P_2,T_2,F_2,W_2,m_{02}]$ be arbitary P/T nets.
A \emph{net morphism} from $N_1$ to $N_2$ is a mapping $\mu: (P_1 \cup T_1) \rightarrow (P_2 \cup T_2)$ such that
$\mu(P_1) \subseteq P_2$, $\mu(T_1) \subseteq T_2$ and  $\forall x,y \in P_1 \cup T_1: W(\mu(x),\mu(y)) = W(x,y) $.
For the initial markings, it holds that $ \forall p_2 \in P_2 : m_{02}(p_2) = \sum_{p_1 \in P_1: (p_1,p_2) \in \mu} m_{01}(p_1)$.
\end{definition}

A net morphism can be extended to a mapping from markings of $ N_1 $ to markings of $ N_2 $ by setting $ m_2(p_2) = \sum_{p_1 \in P_1: (p_1,p_2) \in \mu} m_1(p_1)$ for all $ p_2 \in P_2 $, where $ m_1 \in RS(N_1) $ and $ m_2 \in RS(N_2) $.
A net morphism preserves the reachability between the related nets.

\begin{lemma}[Net Morphism preserves reachability \cite{pinna11}]  \label{prop:morphism}
	Let $ N_1,N_2 $ be two P/T nets, related by a net morphism $ \mu $.
	The transition $ m \xrightarrow{t} m'$ in $N_1$ implies the transition $ \mu(m) \xrightarrow{\mu(t)} \mu(m') $ in $ N_2 $.
\end{lemma}

It is easy to see that $ U $ and $ S $ are related by a net morphism.

\begin{lemma}[Net morphism from unfolding to skeleton \cite{desel91,padberg98}] \label{prop:inducedmorphism}
	Let $ C $ be a coloured Petri net, $ U $ its unfolding and $ S $ its skeleton.
	The mapping $ \mu: (P_{\U} \cup T_{\U}) \rightarrow (P_{\S} \cup T_{\S})$ is a net morphism from $ U $ to $ S $, where
	\begin{itemize}
		\item $ \forall [p,c] \in P_{\U} : \mu([p,c]) = p \in P_{\S}$
		\item $ \forall [t,g] \in T_{\U}$: $\mu([t,g]) = t \in T_{\S}$
	\end{itemize}
\end{lemma}
The net morphism $ \mu $ between $ U $ and $ S $ can as well be extended to the markings of $ U $ and $ S $, such as $ m_{\S}(p) = \sum_{[p,c] \in P_{\U}: ([p,c],p) \in \mu}: m_{\U}([p,c])$ for all $ p \in P_{\S}$, where $ m_{\U} \in RS(U) $ and $ m_{\S} \in RS(S) $.

\medskip
We continue with the introduction of the syntax and semantics of the temporal logic $CTL^*$ \cite{clarke86}.
The foundation for this logic are atomic propositions, properties which are either true or false in a given state.
$ CTL^* $ distinguishes state formulas and path formulas.

\begin{definition}[Syntax of $CTL^*$]
	The temporal logic $ CTL^* $ is inductively defined as follows:
	\begin{itemize}
\itemsep=0.9pt
		\item every atomic proposition is a state formula
		\item if $ \phi $ and $ \psi $ are state formulas, so are $ (\phi \wedge \psi) $, $ (\phi \vee \psi) $, and $ \neg \phi $
		\item every state formula is a path formula
		\item if $ \phi $ and $ \psi $ are path formulas, so are $ (\phi \wedge \psi) $, $ (\phi \vee \psi) $, $ \neg \phi $, $ \mathbf{X	\,} \phi$, $\mathbf{F\,} \phi $, $ \mathbf{G\,}\phi $, $ (\phi \mathbf{\,U\,} \psi) $, $ (\phi \mathbf{\,W\,} \psi) $, and $ (\phi \mathbf{\,R\,} \psi) $
		\item if $ \phi $ is a path formula then $ \mathbf{E\,} \phi $ and $ \mathbf{A\,} \phi $ are state formulas.
	\end{itemize}
\end{definition}

The semantics of $CTL^*$ relies on the concept of paths in the considered system, given as a transition system.

\begin{definition}[Path,Suffix]
	Let $ TS = [Q,q_0,R,A] $ be a transition system.
	A \emph{finite path} starting in state $ q_0 $ is a sequence $ \pi = q_0 \dots q_n $ of states where $ \forall i \in \{0,\dots,n-1\} : (q_i,a, q_{i+1}) \in R $.
	An \emph{infinite path} starting in $ q_0 $ is an infinite sequence $ \pi = q_0 q_1 \dots $ where $ \forall i \in \mathbb{N}: (q_i,a, q_{i+1}) \in R $.
	A path is a finite or infinite path.
	A path is \emph{maximal}, if it is infinite, or is a finite path
	$ q_1 \dots q_n $ where $ q_n $ is a \emph{deadlock}, i.e.~a state where, for all $ q \in Q $, $ (q_n,a,q) \notin R $.
	As a \emph{Suffix} of a path $ \pi $ we define $ \pi_i $ as the part of $ \pi $, starting in $ q_i $.
\end{definition}

The semantics of $ CTL^* $ is defined on infinite paths, as we find them in \emph{Kripke structures}.
A Kripke structure is a transition system $ K = [Q,q_0,R,A,L] $ where $ R $ is total, i.e. every state has at least one successor state.
Thus the maximal paths are always infinite here.
Aditionally, Kripke structures have a labelling function $ L : Q \rightarrow 2^{AP} $, which assigns the set of atomic propositions to every state, which are true in this state.
Every transition system can canonically be transformed into a Kripke structure by adding a \emph{silent transition action} $ (q_d,\tau,q_d) $ to $ R $ for each deadlock state $ q_d $ of the system, which does not have a successor state.
The semantics of $ CTL^* $ is defined by two satisfaction relations, both denoted with $ \models $, that relate markings and state formulas resp.~infinite paths and path formulas according to the following rules.

\begin{definition}[Semantics of $ CTL^* $] \label{def:semantics}
	Let $ K = [Q,q_0,R,A,L] $ be a Kripke structure.
	Let $ q \in Q $ be a state and $ \pi = q_0 q_1 \dots $ an infinite path of the system.
	Let futher be $ \Pi(q) $ the set of paths starting from $ q $.
	The satisfaction of a $ CTL^* $ formula is defined:
	\begin{itemize}
  \itemsep=0.9pt
		\item For an atomic proposition $ \phi $: let $ q \models \phi $ if $ \phi \in L(q) $.
		\item For a state formula $ \phi $: $ \pi \models \phi $, if $ q_0 \models \phi $.
		\item Boolean connectors:
		\begin{itemize}
          \itemsep=0.9pt
			\item $ q \models \neg \phi $, if $ q \not\models \phi $; $ \pi \models \neg \phi $ if $ q_0 \not\models \phi $
			\item $ q \models (\phi \land \psi) $, if $ q \models \phi $ and $ q \models \psi $;
			$ \pi \models (\phi \lor \psi) $ if $\pi \models \phi$ or $ \pi \models \psi $.
		\end{itemize}
		\item Temporal operators:
		\begin{itemize}
        \itemsep=0.9pt
			\item $ \pi \models \mathbf{X\,} \phi $, if $ \pi_1 \models \phi $
			\item $ \pi \models (\phi \mathbf{\,U\,} \psi)$, if $ \exists i \geq 0 : \pi_i \models \psi $ and $ \forall 0 \leq j < i: \pi_j \models \phi $.
		\end{itemize}
		\item Path quantifier: $ q_0 \models \mathbf{E\,} \phi$, if $ \exists \pi \in \Pi(q_0): \pi \models \phi$.
	\end{itemize}
	Let further
	$ \phi \vee \psi $ be equivalent to $ \neg(\neg \phi \wedge \neg \psi) $,
	$ \mathbf{F\,} \phi $ to $ \operatorname{true} \mathbf{\,U\,} \phi $,
	$ \mathbf{G\,} \phi $ to $ \neg \mathbf{\,F\,} \neg \phi $,
	$ \phi \mathbf{\,R\,} \psi $ to $ \neg(\neg \phi \mathbf{\,U\,} \neg \psi) $,
	$ \phi \mathbf{\,W\,} \psi $ to $ \mathbf{G\,} \phi \lor (\phi \mathbf{\,U\,} \psi) $
	and $ \mathbf{A\,} \phi $ to $ \neg \mathbf{\,E\,} \neg \phi $.
\end{definition}

A Kripke structure satisfies a state formula if its initial states does.
It satisfies a path formula if all paths starting in the initial state do.
For $ CTL^* $, several fragments are frequently studied.

\begin{definition}[Fragments of $CTL^*$]
	$ CTL^* $ formula $ \phi $ is in
	\begin{itemize}
\itemsep=0.95pt
	\item $ LTL $ if $ \phi $ does neither contain $ \mathbf{E} $ nor $ \mathbf{A} $
	\item $ ACTL^* $ if $ \phi $ does neither contain $ \mathbf{E} $ nor $ \neg $
	\item $ CTL $ if every occurrence of $ \mathbf{X,U,F,G,R} $ is immediately preceded by an occurrence of $ \mathbf{A} $ or~$\mathbf{E} $
	\item $ ACTL $ if $ \phi $ is in $ ACTL^* $ and $ CTL $
	\item for any fragment $ F $, $ \phi $ is in the fragment $ F_X $ if $ \phi $ is in $ F $ and does not contain $ \mathbf{X} $.
	\end{itemize}
\end{definition}

Since CTL and LTL contain, for all their operators, the dual operator w.r.t~negation, we can push negations to the bottom of formulas. Consequently, $ LTL $ is indeed a subset of $ ACTL^* $.

\section{Relations between reachability graphs} \label{sec:relations}

For describing relations between reachability graphs, we use the concepts of abstraction relation and simulation relation, defined for Kripke structures.

\begin{definition}[Abstraction Relation \cite{milner89}] \label{def:abstraction}
	Let $ K = [Q,q_0,R,A,L] $ and $ \hat{K} = [\hat{Q},\hat{q}_0,\hat{R},\hat{A},\hat{L}] $ be Kripke structures.
	An \emph{abstraction relation} exists between $ K $ and $ \hat{K} $, if there is a surjective abstraction function
	$ \sigma: Q \rightarrow \hat{Q} $, for which it holds that for every $ \hat{q} \in \hat{Q}$ and $ \forall a \in AP : \hat{q} \models a \gdw \forall q \in Q $ with $ (q,\hat{q}) \in \sigma: q \models a $.
\end{definition}

If such an abstraction relation exists between $ K $ and $ \hat{K} $, we say that $ \hat{K} $ (abstract system) abstracts $ K $ (concrete system).
A particular type of abstraction relation is the simulation relation.

\begin{definition}[Simulation Relation \cite{grumberg94}]\label{def:simulation}
	An abstraction relation between $ K = [Q,q_0,R,A,L] $ and $ \hat{K} = [\hat{Q},\hat{q}_0,\hat{R},\hat{A},\hat{L}] $ with the abstraction function $ \sigma: Q \rightarrow \hat{Q} $ is a \emph{simulation relation}, if
	$ \forall q,q_1 \in Q : q \xrightarrow{*} q_1 $ and $ (q,\hat{q}) \in \sigma \rightarrow \exists \hat{q_1} \in \hat{Q}: \hat{q} \xrightarrow{*} \hat{q_1} $ for $ \hat{q} \in \hat{Q} $ and $ (q_1,\hat{q_1}) \in \sigma $.
\end{definition}

If a simulation relation exists between $ K $ and $ \hat{K} $, we say that $ \hat{K} $ simulates $ K $.
$ ACTL^* $ properties are preserved through a simulation relation.

\begin{lemma}[Simulation Relation preserves $ACTL^*$ \cite{clarke86}] \label{prop:simulation}
	Let $ K = [Q,q_0,R,A,L] $ and $ \hat{K} = [\hat{Q},\hat{q}_0,\hat{R},\hat{A},\hat{L}] $ be Kripke structures.
	If there is a simulation relation between $ K $ and $ \hat{K} $, for every $ ACTL^* $ formula $ \phi $, it holds that
	$ \hat{K} \models \phi \Rightarrow K \models \phi $.
\end{lemma}

As deadlocks may occur in Petri nets, a reachability graph is not necessarily a Kripke structure.
To make the concepts of abstraction and simulation formally applicable to Petri nets, we need to transform the reachability graphs into Kripke structures, as described above.
Thus, for every deadlock marking $ m_d $ in a reachability graph, we add a self-loop $ (m_d,\tau,m_d) $ with a silent transition $ \tau $ to $ R $.
From now on, consider the reachability graphs of Petri nets as Kripke structures, arose out of this transformation.
For an abstraction relation, atomic propositions are essential, so we first specify atomic propositions in the context of Petri nets.

\begin{definition}[Atomic proposition] \label{def:atomic}
	Let $ N $ be a Petri net.
	An  \emph{atomic proposition} is an expression $ k_1 p_1 + \dots k_n p_n \leq k $, for some $n \in \mathbb{N}$ with  $k_1, \dots, k_n,k \in \mathbb{Z}$, and $p_1, \dots, p_n \in P$, where $ P $ is the set of places of $ N $.
	A marking $ m $ of a P/T net \emph{satisfies} the proposition $ k_1 p_1 + \dots + k_n p_n \leq k $, iff the term
	$\sum_{i=1}^n k_i \cdot m(p_i)$  evaluates to a number less than or equal to $ k $.
	A marking $ m $ of coloured net \emph{satisfies} proposition  $k_1 p_1 + \dots + k_n p_n \leq k$, iff the term
	$\sum_{i=1}^n k_i \cdot \sum_{c \in \chi(p_i)} m(p_i)(c)$ evaluates to a number less than or equal to $ k $.
	For both, $ m \models a $ denotes the fact that $ m $ satisfies atomic proposition $ a $.
\end{definition}

The net morphism $ \mu: (P_{\U} \cup T_{\U}) \rightarrow (P_{\S} \cup T_{\S}) $ between the unfolding $ U $ of a coloured net $ C $ and its skeleton $ S $ induces an abstraction relation between their reachability graphs.
To show this, we need to specify the unfolding of atomic propositions of coloured nets.
As $ S $ resp. $ C $ normally have another set of places as $ U $, the equisatisfiability between the concrete and the abstract states, required in Definition~\ref{def:abstraction}, is not trivial.

\begin{definition}[Unfolding of Atomic Propositions] \label{def:atomicunfolding}
	Let $ \mu: (P_{\U} \cup T_{\U}) \rightarrow (P_{\S} \cup T_{\S}) $ be the net morphism between $ U $ and $ S $.
	Let $ a_{\C} \in AP_{\C} $ an atomic proposition of a coloured net $ C $.
	Proposition $ a_{\C} $ can be unfolded to an atomic proposition $ a_{\U} \in AP_{\U} $ by substituting every occurrence of any place $ p \in P_{\C} $ by $ \sum_{[p,c] \in P_{\U}: \mu([p,c]) = p} [p,c] $.
\end{definition}

An atomic proposition $ a_{\C} $ and its unfolding $ a_{\U} $ are equisatisfiable.
To make the unfolding of atomic propositions more clear, consider example~\ref{ex:1} and
the atomic proposition $ p \leq 3 $.
As the colour domain of $ p $ is $ \chi(p) = \{g,r,b\}$ and $ p $ would be unfolded to the places $ [p,g],[p,r]$, and $[p,b]$, we unfold the atomic proposition to $ [p,g] + [p,r] + [p,b] \leq 3 $.

With the definition of atomic propositions and their unfolding we can build an abstraction relation between the unfolding of a coloured net and its skeleton.

\begin{theorem}[Abstraction Relation between $ U $ and $ S $] \label{prop:abstractionUS}
	Let $ U $ and $ S $ be related by the net morphism $ \mu: (P_{\U} \cup T_{\U}) \rightarrow (P_{\S}\cup T_{\S}) $ from proposition~\ref{prop:inducedmorphism}.
	The extension of $ \mu $ on the markings of $ U $ and $ S $ yields to a surjective abstraction function $ \sigma $ with $ (m_{\U},m_{\S}) \in \sigma $ for $ (m_{\U},m_{\S}) \in \mu $.
	Therefore, an abstraction relation between the markings of $ U $ and $ S $ exists.
\end{theorem}
It is worth mentioning that markings here include reachable and non-reachable markings.

\begin{proof}
	Let $ a_{\U} \in AP_{\U}$, $a_{\C} \in AP_{\C}$ and $a_{\S} \in AP_{\S}$ be atomic propositions.
	The relation $ \sigma $ is an abstraction relation indeed, if for a marking $ m_{\S} $ of $ S $, $ m_{\S} \models a_{\S} \gdw \forall m_{\U} \models a_{\U} $ with $ (m_{\U},m_{\S}) \in \sigma $.
	If $ m_{\S} \models a_{\S} $, then $ \sum_{i=1}^n k_i \cdot m_{\S}(p_i) \leq k $.
	For every corresponding marking $ m_{\C} $ of $ C $, it holds that $ m_{\C} \models a_{\C} $, as $m_{\S}(p_i) = \sum_{c \in \chi(p_i)} m_{\C}(p_i)(c) $ for every $ i \in \{1,\dots,n\} $ and so, $ \sum_{i=1}^n k_i \cdot \sum_{c \in \chi(p_i)} m_{\C}(p_i)(c) \leq k $.
	Notice that $ m_{\C} $ may be unreachable.
	As the corresponding markings of the unfoldings are equisatisfiable, for every  $ m_{\U} $ of $ U $, it holds that $ m_{\U} \models a_{\U}$.
	Reversed, it must hold that if for a marking  $ m_{\S} $ with $ m_{\S} \not \models a_{\S} $, there is a marking $ m_{\U} $ with $ (m_{\U},m_{\S}) \in \sigma $, for which it holds that $ m_{\U} \not \models a_{\U} $.
	Let $ \sum_i^n k_i \cdot m_{\S}(p_i) > k $.
	For the marking $ m_{\C} $ also holds that $ m_{\C} \not \models a_{\C} $.
	This $ m_{\C} $ might be unreachable again.
	We can see, that for $ m_{\U} $, $ m_{\U} \not \models a_{\U} $ as well.
\end{proof}

The existence of an abstraction relation is not sufficient for transferring the validity results on the markings of $ S $ to $ U $.
We need in fact a simulation relation.
A simulation $ \sigma $ requires the preservation of the transitions between the simulating systems, so it should hold that $ \forall m_{\U},m_{\U\1} \in  RS(U): m_{\U} \xrightarrow{*} m_{\U\1}$ and $ (m_{\U},m_{\S}) \in \sigma \Rightarrow \exists m_{\S\1} \in  RS(S): m_{\S} \xrightarrow{*} m_{\S\1} $ and $ (m_{\U\1},m_{\S\1}) \in \sigma $ for $ m_{\S} \in RS(S) $.
As the coloured net may have deadlocks, which are not preserved in the skeleton as shown in the opening example, there may be silent transitions at the deadlock states of $U$, which are not preserved in the skeleton $S$.
Let $ m_{\U} \in  RS(U)$ be a deadlock of $ U $ not preserved in $S$, so $ m_{\U} \xrightarrow{\tau} m_{\U} $.
Let $ m_{\S} \in  RS(S) $ be the corresponding marking of $ S $ with $ (m_{\U},m_{\S}) \in \sigma $.
Since $ m_{\S} $ is not a deadlock, there is no silent transition added for $ m_{\S} $ and consequentially, there is no marking $ m_{\S\1} \in  RS(S)$ with $ m_{\S} \xrightarrow{*} m_{\S\1} $ and $ (m_{\U},m_{\S\1}) \in \sigma $, as $ m_{\U},m_{\S\1} $ do not fulfill atomic propositions equally.

\section{Simulation relation between reachability graphs} \label{sec:sim}

In this section, we discuss the existence of a simulation relation between the unfolding and the skeleton under various conditions.
As mentioned above, deadlocks that are not preserved in the skeleton, may cause problems.
We therefore distinguish coloured nets, where
\begin{enumerate}
\itemsep=0.9pt
	\item[a)] no deadlocks occur at all (Section~\ref{subsec:free}),
	\item[b)] all deadlocks are preserved in the skeleton (Section~\ref{subsec:findlow}),
	\item[c)] deadlocks are not always preserved. (Section~\ref{subsec:inject})
\end{enumerate}
The kind of the $ ACTL^* $ formula is significant as well. $ ACTL^* $ safety properties permit the use of the skeleton approach even if deadlocks are not preserved, as shown in Section~\ref{subsec:safety}.

\subsection{Deadlock-free nets} \label{subsec:free}

When the net $ C $ resp.~$ U $ has no deadlocks, the net morphism directly leads to a simulation relation between the markings of $ U $ and $ S $.
There is no need to add silent transitions in $ U $ that are not preserved in $ S $.

\begin{theorem}
	Let $ C $ be a coloured net without deadlocks, $ U $ its unfolding and $ S $ its skeleton.
	The net morphism $ \mu: (P_{\U} \cup T_{\U}) \rightarrow (P_{\S}\cup T_{\S}) $ from Lemma~\ref{prop:inducedmorphism} induces a simulation relation between the markings of $ U $ and $ S $.
\end{theorem}

\begin{proof}
	The reachability graph $ R_{\U}(m_0) $ is a Kripke structure, without adding silent transitions.
	Let $ m_{\U}, m_{\U\1} \in RS(U)$ and $m_{\S} \in RS(S) $ be the corresponding marking of $ m_{\U}  $ with $ (m_{\U},m_{\S}) \in \mu $.
	The markings $ m_{\U} $ and $ m_{\S} $ are then related by the abstraction relation $ \sigma $ from Theorem\ref{prop:abstractionUS}: $ (m_{\U},m_{\S}) \in \sigma$.
	Because net morphisms preserve reachability, for $ t_{\U} \in T_{\U}$, if it holds that if $m_{\U} \xrightarrow{t_{\U}} m_{\U\1} $ in $ R_{\U}(m_0) $, then there is a marking $ m_{\S\1} \in RS(S)$ with $ m_{\S} \xrightarrow{\mu(t)} m_{\S\1} $ in $ R_{\S}(m_0) $, for which $ (m_{\U\1},m_{\S\1}) \in \mu $.
	Consequently, for all markings $ m_{\U},m_{\U\1} \in RS(U)$ with $ m_{\U} \xrightarrow{*} m_{\U\1} $ and $ (m_{\U},m_{\S}) \in \sigma $, there is a marking $ m_{\S\1} \in RS(S)$ with $ m_{\S} \xrightarrow{*} m_{\S\1} $ in $ R_{\S}(m_0) $ and $ (m_{\U\1},m_{\S\1}) \in \sigma $.
\end{proof}

Thus, according to proposition~\ref{prop:simulation}, $ ACTL^* $ properties are preserved.
If we can guarantee, that the considered net is deadlock-free, the skeleton abstraction can be used for transferring positive results of an $ ACTL^* $ verification in $ S $ to $ U $.

\subsection{Deadlock preservation} \label{subsec:findlow}

We now consider the case where a Petri net has deadlocks.
The reachability graph of this net is not readily a Kripke structure, hence all deadlock states were extended with a self loop with a silent action.
In \cite{findlow92}, two necessary and sufficient criteria are formulated, defining a class of coloured Petri nets which have a \emph{deadlock-preserving skeleton}.

\medskip
This means that every dead marking of the coloured net has a dead skeletal image, thus no deadlock of the coloured net is invisible in the skeleton.
By that it is possible to detect all deadlocks just by the skeletal analysis.
The two criteria have both characteristic advantages for our skeleton-based analysis and verification.
The first criterion relates to an equivalence relation of the transitions of the coloured net. The second one concerns the liveness of markings.
To determine, whether a given coloured Petri net has a deadlock-preserving skeleton, it is appropriate to use the first criterion, as it refers exclusively to the net structure and does not consider the behaviour of the net.

As we assume that the coloured net $C = [P_{\C},T_{\C},F_{\C},W_{\C},\chi,\gamma,m_{0\C}]$ is uniform, the number of input tokens a transition $ t \in T_{\C} $ requires from each place $ p_i $ for $ \scopen{i} $ with $ n = \abs{P_{\C}} $ is unambiguous.
From now on, this number of input tokens for a transition $ t $ is denoted as $ f_i(t) $, where $ f_i(t) = \abs{W_{\C}(p_i,t)} $ for $ \scopen{i} $.
These numbers form an input vector $ f : T \rightarrow \mathbb{N}^n$ for every transition $ t \in  T_{\C} : f(t) = (f_1(t),f_2(t),\dots,f_n(t)) $.
Building on that, we can determine a preorder $ (T_{\C},\lesssim) $ of the transitions of $ C $, such that $ \forall t,t' \in T_{\C}: t \lesssim t'$, iff $ f(t) \leq f(t') $.
This leads to an equivalence relation $ \sim $ on $ T_{\C} $, such that $ \forall t,t' \in T_{\C}: t \sim t'$, iff $ f(t)=f(t') $.
Transitions with an identical input are aggregated in one equivalence class of this equivalence relation.
Let $ T_{\C}/\!_\sim $ be the set of equivalence classes of $ T_{\C} $.
The preorder $ (T_{\C},\lesssim) $ induces a partial order $ (T_{\C}/\!_\sim, \leq) $ on the equivalence classes, such that $ \forall [t], [t'] \in T_{\C}/\!_\sim : [t] \leq [t']$, iff $t \lesssim t'$.

\begin{definition}[Full Transition Class] \label{def:full}
	An equivalence class $ [t] \in  T_{\C}/\!_\sim $ is \emph{full}, if for every marking $ m_{\C} $ of $ C $ with $ \abs{m_{\C}(p_i)} = f_i(t) $ for all $ \scopen{i} $, there is a transition $ t \in [t] $ that is enabled in $m_{\C}$.
\end{definition}

In other words, $ [t] $ is full, if any collection of bags matching the input size requirements of $ [t] $ also matches the input colour distribution requirements of one transition $ t \in [t] $ of this equivalence class.
This leads to the following lemma:

\begin{lemma}[Deadlock-Preserving Skeleton \cite{findlow92}]
	Let $ C $ be a uniform, coloured Petri net.
	$ C $ has a deadlock-preserving skeleton, iff every minimal transition class of $ C $ regarding $ (T_{\C}/\!_\sim, \leq) $ is full.
\end{lemma}

This is a necessary and sufficient condition.
With this criterion, we can show that the simulation relation between a coloured net and its skeleton is preserved for a subclass of coloured nets, which have a deadlock-preserving skeleton.
If the criterion holds, the net morphism $ \mu: (P_{\U} \cup T_{\U}) \rightarrow (P_{\S}\cup T_{\S}) $ induces a simulation relation.
If a coloured net has a deadlock-preserving skeleton, for every added silent transition for a dead marking $ m_{\U} \in RS(U) $ of $ U $, there is an added silent transition for the dead skeletal image $ m_{\S} \in RS(S) $ with $ (m_{\U},m_{\S}) \in \mu $ as well.
With regard to Lemma~\ref{prop:simulation}, we can verify the $ ACTL^* $ properties only in $ S $ without risking wrong conclusions about the behavior of $ C $.

\subsection{Inject deadlocks to skeleton} \label{subsec:inject}

The main focus of this section are nets with deadlocks, but without deadlock-preserving skeleton.
Here, the net morphism does not induce a simulation relation, so the $ ACTL^* $ results cannot be transferred directly from the skeleton to the coloured net.
We present an approach to modify the skeleton net such that every deadlock of the unfolding occurs in the new skeleton, but potentially with some delay.
In this case we cannot guarantee that every dead marking has a dead skeletal image, but we can at least guarantee that for a dead marking, the corresponding skeletal deadlock occurs after a finite number of actions.

\begin{definition}[Modified Skeleton Net]
	Let $ C = [P_{\C},T_{\C},F_{\C},W_{\C},\chi,\gamma,m_{0\C}] $ be a uniform coloured net.
	The \emph{modified skeleton} $ S' $ can be constructed from the skeleton $ S $ as, for every preset place $ p \in P_{\C} $ of a non-full minimal transition class $ [t] $, a complement place $ \overline{p} $ and a recipient transition $ t_r $ with $ \bullet t_r = \{p\} $ and $ t_r \bullet = \{\overline{p}\} $ are introduced with $ W(p,t_r) = W(t_r,\overline{p}) = 1 $.
	Apart from that, $ S' $ and $ S $ are identical.
\end{definition}

The modified skeleton has another behaviour than the original skeleton.
Every recipient transition $ t_r $ can successively empty its preset place $ p $ and stores the tokens on the complement place $ \overline{p} $.
These actions can be considered as silent actions of $ S' $.
Once a token is stored on $ \overline{p} $, it cannot leave this place anymore.
So, after a finite number of actions of the recipient transitions, the preset of $ [t] $ is empty and the transitions in $ [t] $ cannot fire anymore.
The deadlock of $ U $ occurs in $ S' $ after a finite number of silent actions of the recipient transitions.
Between $ U $ and $ S' $ a stuttering simulation holds, which is a weakened version of a simulation relation.
The next definition talks about partitions of infinite paths.
A partition of a path $ \pi $ consists of finite or infinite subpaths $ B_i $, such that their concatenation yields the whole $ \pi $.

\begin{definition}[Stuttering Simulation \cite{penczekszretergerthkuiper00}] \label{def:stuttering}
	Let $ K = [Q,q_0,R,A,L] $ and $ \hat{K} = [\hat{Q},\hat{q}_0,\hat{R},\hat{A},\hat{L}] $ be Kripke structures and $ a $ be an atomic proposition. A mapping $ \sigma_s: Q \rightarrow \hat{Q} $ is a \emph{stuttering simulation relation} if the following conditions hold:
	\begin{itemize}
\itemsep=0.9pt
		\item $ (q_{0},\hat{q_0}) \in \sigma_s $
		\item $ (q,\hat{q}) \in \sigma_s \Rightarrow q \models a \gdw \hat{q} \models a $ and for every path $ \pi = q_0 q_1 q_2 \dots $ of $ K $, there is a path $ \hat{\pi} = \hat{q_0} \hat{q_1} \hat{q_2} \dots $ of $ \hat{K} $, such that we can find partitions $ B_0,B_1,B_2,\dots $ for $ \pi $ resp. $ \hat{B_0},\hat{B_1},\hat{B_2},\dots $ for $ \hat{\pi} $ for which holds that:
		\begin{itemize}
\itemsep=0.9pt
			\item $ \forall i \geq 0: B_i,\hat{B_i}$ are not empty and finite
			\item every state of $ \hat{B_i} $ is related with every state of $ B_i $ by $ \sigma_s $.
		\end{itemize}
	\end{itemize}
\end{definition}

If two systems are related by a stuttering simulation, the behaviour of the concrete system $ K $ is simulated by the abstract system $ \hat{K} $, but $ \hat{K} $ can run internal silent actions while simulating.
Between the unfolding and the modified skeleton, we can observe this stuttering simulation.
To prove this, we first need to establish a relation between the markings of $ U $ and $ S' $.
Therefore, we create a relation between the markings of $ S $ and the markings of $ S' $.
A Marking $ m_{\S} $ of $ S $ and a marking $ m_{\S'} $ of $ S' $ are related, if
	\begin{itemize}
\itemsep=0.9pt
		\item $ m_{\S}(p) = m_{\S'}(p) + m_{\S'}(\overline{p}) $ for $ p \in \bullet [t] $, where $ [t] $ is a non-full minimal transition class
		\item $ m_{\S}(p) = m_{\S'}(p) $ otherwise.
	\end{itemize}
The relation between a marking $ m_{\U} $ and a marking $ m_{\S'} $ can then be established by composing the abstraction relation from $ m_{\U} $ to $ m_{\S} $ and with the one just defined.
Thus, the relation between the markings of $ U $ and $ S' $ is an abstraction relation.
The silent actions of the recipient transitions move the tokens of the preset places to their complementary places.
No matter if they have moved one or all tokens, the sum over the places $ p $ and $ \overline{p} $ is always invariant.

\begin{theorem}[Stuttering Simulation between $ U $ and $ S' $]
	Let $ C $ be a uniform coloured net, $ U $ its unfolding and $ S' $ its modified skeleton.
	Between the markings of $ U $ and $ S' $, a stuttering simulation $ \sigma_s $ holds.
\end{theorem}

\begin{proof}
	The definition of the marking guarantees, that an abstraction relation exists between the markings of $ U $ and $ S' $.
	States, which are related by the $ \sigma_s $, fullfil atomic propositions equally.
	Between the initial markings $ m_{0\U} $ and $m_{0\S'} $ the stuttering simulation holds.
	Now consider the path $ \pi_{\U} = m_{1\U}m_{2\U} \dots $ of $ U $ and the correspondig path $ \pi_{\S'} = m_{1\S'}m_{2\S'} \dots $ of $ S' $, where $ (m_{i\U},m_{i\S'}) \in \mu$ for all $ i $.
	The the partitioning of $ \pi_{\U} $ and the corresponding path $ \pi_{\S'} $ in $ S' $ is obtained as follows:
	For a marking $ m_{i\U} $ of path $ \pi_{\U} $, which is not a deadlock, the corresponding part of $ \pi_{\S'} $ is simply $ m_{i\S'} $ with  $ (m_{i\U},m_{i\S'}) \in \sigma_s$.
	The partitioning of the paths for this parts is trivial: $ B_{i\U} = \{m_{i\U}\} $ resp. $B_{i\S'} = \{m_{i\S'}\}$.
	Let now be $ m_{i\U} $ a deadlock, which is only followed by the self-loop-$ \tau $-actions.
	The corresponding marking $ m_{i\S'} $ is not necessarily a deadlock.
	Firing the recipient transitions in $ m_{i\S'} $ yields to a sequence $ \tau^* $ which ends in a deadlock marking $ m_{d\S'} $, where only the self-loop-$ \tau $-action is possible as well.
	For partitioning, $ B_{i\U} $ contains only the deadlock state $ m_{i\U} $ of $ U $.
	$ B_{i\S'} $ contains the states $ m_{i\S'}, m_{i+1\S'},m_{i+2\S'}, \dots ,m_{d\S'}$, where $ m_{i+1\S'},m_{i+2\S'}, \dots $ are the markings, reached by actions of the recipient transitions and $ m_{d\S'} $ is the delayed deadlock marking.
	All states in $ B_{i\S'} $ have the same validity of atomic propositions and so they can be related with $ m_{i\U} $ by $ \sigma_s $.
	So, between $ U $ and $ S' $ a stuttering simulation holds.
\end{proof}

A stuttering simulation preserves $ACTL^*_X $ properties.

\begin{lemma}[Stuttering simulation preserves $ACTL^*_X$ \cite{penczekszretergerthkuiper00}]
	\label{prop:stuttering}
	Let $ K = [Q,q_0,R,A,L] $ and $ \hat{K} = [\hat{Q},\hat{q}_0,\hat{R},\hat{A},\hat{L}] $ be Kripke structures, which are related by a stuttering simulation.
	Then $ \hat{K} \models \phi \Rightarrow K \models \phi $ for any $ ACTL^*_X$ formula $ \phi $.
\end{lemma}

$ ACTL^*_X $ formulas permit claims about the overall behaviour of the system, except for referring to next states.
The silent actions in the abstract system can generate new next states, which replace the actual simulating next state.
Because of this, assumptions on next states can be falsified, which explains the restriction to $ ACTL^*_X $.
Nevertheless, the validity of at least a subset of $ ACTL^* $ formulas can be transferred from the modified skeleton to the unfolding.

\subsection{Safety properties} \label{subsec:safety}
In the context of net morphisms, safety properties make an exception with regard of their validation.

\begin{definition}[Safety Property \cite{katzgrumberggeist99}]
	An $ACTL^*$ property is a \emph{safety property}, if only the temporal operators $ \mathbf{W}, \mathbf{X}$ and the path quantifier $ \mathbf{A} $ occur.
\end{definition}

We claim that a safety property $ \phi $ is preserved by a net morphism even if that morphism does not induce a simulation relation.
In the context of the skeleton abstraction, the abstraction relation between the markings of $ U $ and $ S $ is sufficient for the preservation of $ ACTL^* $ safety formulas.
This fact was already informally mentioned in \cite{padberg98}.
However, that paper did not precisely define the class of properties and did not
prove the claim.

\begin{theorem}[Net Morphisms preserve $ ACTL^* $ Safety Properties]
	Let $ C $ be a coloured net, $ U $ its unfolding and $ S $ its skeleton.
	Let $ \mu: (P_{\U} \cup T_{\U}) \rightarrow (P_{\S}\cup T_{\S}) $ a net morphism, $ \phi_{\S} $ an $ ACTL^* $ safety property and $ \phi_{\U} $ its unfolding after Definition~\ref{def:atomicunfolding}.
	Let $ m $ be a marking of $ U $.
	Then it holds that: $ \mu(m) \models \phi_{\S} \Rightarrow m \models \phi_{\U}$.
\end{theorem}

\begin{proof}
	We prove the contraposition $ m \not \models \phi_{\U} \Rightarrow \mu(m) \not \models \phi_{\S} $ by induction on the structure of $ \phi_{\C} $. \\
	\emph{Base:}
	If $ m \not \models \phi_{\U} $, then $ \mu(m) \not \models \phi{\S} $, corresponding to Definition~\ref{def:abstraction}. \\
	\emph{Step:}
	We therefore distinguish between the possible structures of $ \phi_{\U} $ and $ \phi_{\S} $:
	\begin{enumerate}
\itemsep=0.9pt
		\item $ \phi_{\U} = \psi_{\U} \land \xi_{\U} $ resp. $ \phi_{\U} = \psi_{\U} \lor \xi_{\U} $: the induction hypothesis can directly be applied to $ \psi_{\U} $ and $ \xi_{\U} $.
		\item $ \phi_{\U} = \mathbf{A} \psi_{\U} $: If $ m \not \models \phi_{\U} $, there is a path $ \pi = m \, m_1 m_2 \dots $ with $ \pi \not \models \psi_{\U} $.
		Because reachability is preserved, there is a path $ \mu(\pi) = \mu(m) \mu(m_1) \mu(m_2) \dots $ with $ \mu(\pi) \not \models \psi_{\S} $.
		So, $ \mu(m) \not \models \mathbf{A} \psi_{\S} $ resp. $ \mu(m) \not \models \phi_{\S}$.
		\item $ \phi_{\U} = \mathbf{A} \mathbf{X} \psi_{\U} $:
		If $ m \not \models \phi_{\U} $, there is a path $ \pi = m \, m_1 m_2 \dots $ where $ m_1 \not \models \psi_{\U} $.
		For the skeleton, there is a path $ \mu(\pi) = \mu(m) \mu(m_1) \mu(m_2) \dots $ where $ \mu(m_1) \not \models \psi_{\S} $.
		So, $ \mu(m) \not \models \mathbf{A} \mathbf{X} \psi_{\S} $ resp. $ \mu(m) \not \models \phi_{\S}$.
		\item $ \phi_{\U} = \mathbf{A} \psi_{\U} \mathbf{W} \xi_{\U} $, which is the disjunction between
		\begin{enumerate}
\itemsep=0.9pt
			\item[a)] $ \phi_{\U} = \mathbf{A} \mathbf{G} \psi_{\U} $:
			If $ m \not \models \mathbf{A} \mathbf{G} \psi_{\U} $, there is a path $ \pi = m \, m_1 m_2 \dots $ with a marking $ m_i \not \models \psi_{\U} $.
			Hence, $ \pi \not \models \mathbf{G} \psi_{\U} $.
			Again, the preservation of reachability leads to a path $ \mu(\pi) = \mu(m) \mu(m_1) \mu(m_2) \dots $ with a marking $ \mu(m_i) \not \models \psi_{\S} $.
			So, $ \mu(\pi) \not \models \mathbf{G} \psi_{\S} $ and thus $ \mu(m) \not \models \mathbf{A} \mathbf{G} \psi_{\S} $;
			\item[b)] $ \phi_{\U} = \mathbf{A} \psi_{\U} \mathbf{U} \xi_{\U} $:
			If $ m \not \models \mathbf{A} \psi_{\U} \mathbf{U} \xi_{\U} $ there is a path $ \pi = m \, m_1 m_2 \dots $ with $ \pi \not \models (\psi_{\U} \mathbf{U} \xi_{\U})$.
			This is possible in two different ways:
			On the one hand, for all $ i \geq 0:m_i \not \models \xi_{\U} $ can hold, hence $ \pi \not \models \mathbf{G} \xi_{\U} $.
			This can be treated analogously to case 4.a).
			On the other hand, there might be a $ m_i \models \xi_{\U} $, but there is also a $ m_j $ with $ j < i $ and $ m_j \not \models \psi_{\U} $.
			Then, there is a path $ \mu(\pi) = \mu(m) \mu(m_1) \mu(m_2) \dots $ with $ \mu(m_i) $ and $ \mu(m_j) $ with $ \mu(m_i) \models \xi_{\S}$ and $ \mu(m_j) \not \models \psi{\S} $ as well.
			Hence, it holds $ \mu(m) \not \models \mathbf{A} \psi_{\S} \mathbf{U} \xi_{\S}$.
		\end{enumerate}
		In both cases, $ \mu(m) \not \models \phi_{\S}$.
	\end{enumerate}

\eject

The invalidity of $ \phi_{\U} $ can always be proven with a finite counterexample path.
Deadlocks may just occur in the last marking of this path.
Let $ m_i $, the last marking of the counterexample were we can see the invalidity of $ \phi_{\U} $, be a deadlock.
Because we consider Kripke structures, every path of the system is infinite.
The counterexample path $ \pi $ is therefore continued to an infinite path $ \pi = m \, m_1 m_2 \dots m_i m_i m_i \dots $ by repeating the deadlock state $ m_i $.
This repetition does not change the finiteness of the counterexample.
If the deadlock $ m_i $ is preserved in the skeleton, this leads to an corresponding path $ \mu(\pi) $ with repetitions as well: $ \mu(\pi) = \mu(m) \mu(m_1) \mu(m_2) \dots \mu(m_i) \mu(m_i) \mu(m_i) \dots $ .
The invalidity of $ \phi_{\S} $ remains unchanged.
If the deadlock is not preserved, the path $ \mu(\pi) $ has another sequel: $ \mu(\pi) = \mu(m) \mu(m_1) \mu(m_2) \dots \mu(mi) \mu(m_{i+1}) \mu(m_{i+2}) \dots $ with $ \mu(m_{i+1}) \neq \mu(m_i)$.
The counterexample is transferred exactly up to and including $ m_i $, the markings $ \mu(m_{i+1}) \mu(m_{i+2}) \dots $ do not change the invalidity of $ \phi_{\S} $.
\end{proof}

\section{Checking full transition classes in symmetric nets}
\label{sec:fullnesscheck}

We proceed with an algorithmic approach to check whether a transition class is full. A brute-force solution would be to enumerate all firing modes of the transitions in the class and to check whether these firing modes cover all distributions of colors on their pre-places as stated in Definition~\ref{def:full}.
This approach may be very inefficient, as already observed in
\cite{iddunfolding}.
It would in particular prevent the application of the skeleton approach to colored nets that have an unfolding too large to be constructed.
Following the approach of \cite{iddunfolding}, we rather create an automaton that accepts precisely those assignments to the variables on the arcs to resp. from transition $ t $, which are firing modes of $ t $, i.e. that satisfy the guard $ \gamma(t) $.
We therefore assume the set of all variables to be ordered.
Then an assignment to the variables $ x_{i_1}, \dots, x_{i_n} $with $x_{i_j} < x_{i_k}$, for $j < k$, is a sequence of length $n$ and the $j$-th element of the sequence is in the domain of variable $x_{i_j}$.
A set of assignments to some set of variables is a set of sequences, all having same length.
The domians of the variables serve as alphabet for these sequences.
We adapt the concept of classical finite automata to this setting.

\begin{definition}[Automaton] \label{def:automata}
	Given a coloured net $ C $, a finite automaton $[X, Q, q_0, \delta, F]$ consists of a set $ X $ of all variables occurring in $ C $, a finite set $ Q $ of states, an initial state $ q_0 \in Q $,
	a set $ F $ of final states ($ F \subseteq Q $), and a deterministic transition function $ \delta: Q \times D \rightarrow Q $, where $ D $ is the union of all domains for the variables in $ X $.
\end{definition}

With this definition, we permit an assignment of the variables of $ X $ with elements from the domain $ D $ as an input for our automaton.
The domain for a variable $ x $ on arc $ F_{\C}(p,t) $ resp. $ F_{\C}(t,p) $ is the colour domain $ \chi(p) $ of place $ p $.
As usual, a sequence is accepted if a run starting in $ q_0 $ with this sequence ends in a final state.

In the sequel, we show how to construct an automaton that accepts the enabled firing modes of a transition. The resulting automaton can then be used to represent all distributions of
tokens that can be consumed by this transition. This information can finally be combined for all transitions of a transition class for checking whether it is full.
For constructing the automaton, we first need to consider the arc inscriptions and the guards more precisely.
We therefore choose the syntax of \emph{symmetric high level nets}, also known as stochastic well formed nets \cite{stochwellformed} although our approach may be easily adapted to other dialects of high-level nets.
The syntax of symmetric nets is simple and reasonably formalized in the PNML standard \cite{pnml}.
All high-level nets in the yearly model checking contests \cite{mcc2021} are modeled as symmetric nets.
Until now, we only suppose the arc inscriptions as a finite set of variables and the guard as a boolean predicate, which can be evaluated to true or false.

In detail, for symmetric nets, arc inscriptions are formal sums of terms or tuples of terms.
They have a rather restricted syntax for terms.

\begin{definition}[Term]
	\label{def:terms}
	Let $ X $ be a set of variables and $ \mathbb{Z} $ the set of integers.
	A \emph{term} can be
	\begin{itemize}
\itemsep=0.9pt
		\item a variable $ x \in X $,
		\item a constant $ k \in \mathbb{Z}$,or
		\item an increment term $ \Tau++ $ or decrement term $ \Tau-- $, for a term $ \Tau $.
	\end{itemize}
\end{definition}

Given an assignment $ \alpha $ to the variables in $ X $, the semantics $ \operatorname{val} $ of a term is defined by the conditions $ \operatorname{val}(x,\alpha) = \alpha(x) $, $\operatorname{val}(k,\alpha) = k $, $ \operatorname{val}(\Tau++,\alpha) = \operatorname{val}(\Tau,\alpha) +1 $, and $ \operatorname{val}(\Tau--,\alpha) = \operatorname{val}(\Tau,\alpha) -1 $.
Addition and subtraction is supposed to be modulo the boundaries of the domain.
Terms may occur positive or negative in the formal sum of an arc incription.

In symmetric nets, a guard consists of expressions, which are basically a comparison of terms of their boolean connections.

\begin{definition}[Expression]
	\label{def:expression}
	An \emph{expression} can be a comparison $\Tau_1 \oplus \Tau_2$ ($\oplus \in \{<,>,\leq,\geq,=, \neq\}$ between two terms $\Tau_1$ and $\Tau_2$,
	or a Boolean combination of expressions. We use the standard semantics for all operators.
\end{definition}

We only consider conjunction and disjunction as Boolean operators since negation can be removed using de Morgan's rules and the set of comparisons is closed under negation.

\subsection{Simplification of arc inscriptions}

Before constructing the automata, we present a simplification for the arc inscriptions, that turns the formal sum of terms resp. tuples of terms into a formal sum of simple variables resp. tuples of simple variables such that every variable occurs only once in any arc connected to a transition.

In general, assume an arc inscription $ \Tau_1 + \dots + \Tau_p - \Tau'_1 - ... - \Tau'_n $ with $ p $ positve and $ n $ negative terms.
This formal sum has size $ m = p - n $ with $ m > 0 $ (otherwise, the arc inscription does not make sense).
We introduce $ m $ variables $ x_1, \dots, x_m $.
The number of variables is the same as the number of token which pass this arc.

\medskip
The arc inscription is replaced with the (positive) formal sum $ x_1 + \dots + x_m $ of the fresh variables and the guard $ \gamma(t) $ of the transition $ t $ is extended to
$$ \gamma(t) \wedge \bigvee_{\pi \in \text{permutations}(p)}\big(\bigwedge_{i=1}^{m} x_i = \Tau_{\pi(i)} \wedge \bigwedge_{j=1}^n \Tau'_{j} = \Tau_{\pi(m+j)}\big). $$

Each permutation $ \pi \in \text{permutations}(p) $ is a bijection mapping variables and negative terms to occuring the positive terms here.

The combinatorics in this construction reflects the fact that tokens on places are not ordered.
This fact can be ignored in subsequent constructions.
The combinatorics furthermore reflects the fact that negative terms in a formal sum must match some positive term since the resulting multiset cannot contain negative multiplicities.
It is obvious that the modification does not change the semantics of the net.

\begin{example}
	If an arc inscription has the shape $ \Tau_1 + \Tau_2 + \Tau_3 - \Tau_4 $, we introduce two fresh variables $ x_1 $ and $ x_2 $ (as the size of the formal sum is 2). We then replace the arc inscription with $ x_1 + x_2 $ and extend the guard $ \gamma(t) $ of transition $ t $ to
	$$\gamma(t) \wedge ((x_1 = \Tau_1 \wedge x_2 = \Tau_2 \wedge \Tau_4 = \Tau_3) \vee (x_1 = \Tau_2 \wedge x_2 = \Tau_1 \wedge \Tau_4 = \Tau_3) \vee (x_1 = \Tau_1 \wedge x_2 = \Tau_3 \wedge \Tau_4 = \Tau_2)$$
	$$\vee (x_1 = \Tau_3 \wedge x_2 = \Tau_1 \wedge \Tau_4 = \Tau_2) \vee (x_1 = \Tau_2 \wedge x_2 = \Tau_3 \wedge \Tau_4 = \Tau_1) \vee (x_1 = \Tau_3 \wedge x_2 = \Tau_1 \wedge \Tau_4 = \Tau_1)).$$
\end{example}

If tuples appear in arc inscriptions, we replace them by a tuple of variables instead of a single variable. For instance, arc inscription $ <\Tau_1, \Tau_2, \Tau_3>$ is replaced with $ <x_1,x_2,x_3> $ and the guard $ \gamma(t) $ is extended to $ \gamma(t) \wedge x_1 = \Tau_1 \wedge x_2 = \Tau_2 \wedge x_3 = \Tau_3$.
This way, we may reduce all future considerations to variables that represent basic domains (color sets that are not cross products of other color sets).
Tuple variables in the guard itself are replaced accordingly.
Again, the semantics of the net is preserved by the modification.
From now on, consider the arc inscriptions as simplified.

In a typical symmetric net, formal sums are small, so the introduced combinatorics is moderate.
There is one exception, though.
Beyond the formal sums considered so far, symmetric nets permit some $ all(\chi(p)) $ construct that represents the formal sum of all elements of color domain $ \chi(p)$ of a place $ p \in P_{\C} $.
Since the combinatorics introduced by that construct is intractable, we disregard it.
In our implementation, all transition classes where an $ all $ construct is used in an arc inscription, are treated as if they were not full. This way, correctness of our approach is not at stake.

The following subsections show how to represent terms and expressions as automata.
This way, we obtain an automaton that accepts all firing modes of a transition.
We map that automaton to token distributions on the pre-places and finally aggregate the resulting automata for checking whether a transition class is full.

\subsection{Term and expression automata}

A term basically represents a value that may depend on an assignment to its occurring variables.
In the field of symmetric coloured level nets this value is an element of the color domain of the connected place.
Color domains in symmetric nets are enumerations or intervals of integer numbers.
Since enumerations can be coded as integers, we shall treat all domains as integer intervals.
Thanks to the simplifications in the previous section, we may disregard cross-product domains.

We represent a term $ \Tau $ as a term automaton $ A_\Tau $.
A term automaton extends automata as in Definition\ref{def:automata} with a mapping $ V: F \rightarrow D $ which maps an element from the domains to every final state.
The idea is that, for a sequence representing assignment $ \alpha $, the reached final state $ q_f $ satisfies $ V(q_f) = \operatorname{val}(\Tau,\alpha) $.
The construction itself is rather obvious, so we reduce our presentation to a few examples, shown in Figure \ref{fig:term}.
Regarding Figure \ref{fig:term3}, remind that incrementation is interpreted modulo the domain size, so the bottom right state is indeed $ q_{c_1} $.

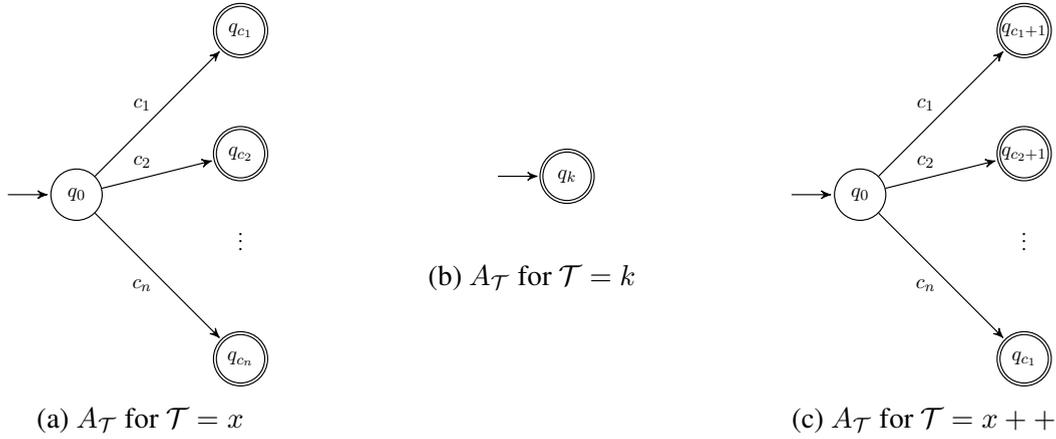
\begin{figure}[ht]
\vspace*{-3mm}
 \resizebox{\textwidth}{!}{
	\begin{subfigure}{0.3\textwidth}
    \hspace{-1cm}
    \begin{tikzpicture}[->,>=stealth',shorten >=1pt,auto,scale = 1,every node/.style={scale=0.65,minimum size=2.5cm,inner sep=.3mm},initial text={}]
        \node[state,initial] (q0) at (0,0) {$q_0$};
        \node[state,accepting] (q1) at (2,2) {$q_{c_1}$};
        \node[state,accepting] (q2) at (2,0.5) {$q_{c_2}$};
        \node[state,accepting] (qn) at (2,-2) {$q_{c_n}$};

        \node (dots) at (2,-0.5) {$ \vdots $};

        \node (c1) at (0.8,1.1){$c_1$};
        \node (c2) at (0.8,0.4){$c_2$};
        \node (c3) at (0.8,-1.1){$c_n$};

        \path (q0) edge[pos=0.6] node {} (q1);
        \path (q0) edge[pos=0.6] node {} (q2);
        \path (q0) edge[pos=0.35] node {} (qn);
      \end{tikzpicture}
      \subcaption{$ A_\Tau $ for $ \Tau = x $}
      \label{fig:term1}
  \end{subfigure}
  \begin{subfigure}{0.3\textwidth}
      \hspace{0.2cm}
      \begin{tikzpicture}[->,>=stealth',shorten >=1pt,auto,scale = 1,every node/.style={scale=0.65,minimum size=2.5cm,inner sep=.3mm},initial text={}]
        \node[state,initial,accepting] (qk) at (0,0) {$q_k$};
      \end{tikzpicture}
      \subcaption{$ A_\Tau $ for $ \Tau = k $}
      \label{fig:term2}
  \end{subfigure}
  \begin{subfigure}{0.3\textwidth}
      \hspace{-1cm}
      \begin{tikzpicture}[->,>=stealth',shorten >=1pt,auto,scale = 1,every node/.style={scale=0.65,minimum size=2.5cm,inner sep=.3mm},initial text={}]
        \node[state,initial] (q0) at (0,0) {$q_0$};
        \node[state,accepting] (q1) at (2,2) {$q_{c_1+1}$};
        \node[state,accepting] (q2) at (2,0.5) {$q_{c_2+1}$};
        \node[state,accepting] (qn) at (2,-2) {$q_{c_1}$};

        \node (dots) at (2,-0.5) {$ \vdots $};

        \node (c1) at (0.8,1.1){$c_1$};
        \node (c2) at (0.8,0.4){$c_2$};
        \node (c3) at (0.8,-1.1){$c_n$};

        \path (q0) edge[pos=0.6] node {} (q1);
        \path (q0) edge[pos=0.6] node {} (q2);
        \path (q0) edge[pos=0.35] node {} (qn);
      \end{tikzpicture}
      \subcaption{$ A_\Tau $ for $ \Tau = x++ $}
      \label{fig:term3}
  \end{subfigure}
  }\vspace*{-2mm}
  \caption{Examples of term automata. The domain for variable $ x $ is $ \{c_1,\dots,c_n\} $ and $ k $ is a constant.}
  \label{fig:term}
\end{figure}

\begin{figure}[!h]
\vspace*{-9mm}
  \resizebox{\textwidth}{!}{
  \begin{subfigure}[b]{0.53\textwidth}
  \hspace{-0.8cm}
      \begin{tikzpicture}[->,>=stealth',shorten >=1pt,auto,scale = 1,every node/.style={scale=0.65,minimum size=2.5cm,inner sep=.3mm},initial text={}]
        \node[state,initial] (q0) at (0,0) {$q_0$};
        \node[state,accepting] (q1) at (4,2) {$q_{c_1}$};
        \node[state,accepting] (q2) at (4,0.5) {$q_{c_2}$};
        \node[state,accepting] (qn) at (4,-2) {$q_{c_n}$};

        \node[state] (r) at (2,0) {$r$};

        \node (dots) at (4,-0.5) {$ \vdots $};

        \path (q0) edge[pos=0.6] node {} (r);

        \node (d) at (1,0.35){$[d_1,d_n]$};

        \path (r) edge[pos=0.6] node {} (q1);
        \path (r) edge[pos=0.6] node {} (q2);
        \path (r) edge[pos=0.35] node {} (qn);

        \node (c1) at (2.8,1.1){$c_1$};
        \node (c2) at (2.8,0.4){$c_2$};
        \node (c3) at (2.8,-1.1){$c_n$};
    \end{tikzpicture}
    \subcaption{Insert variable less than $x$ to $ A_\Tau $ in Figure \ref{fig:term1}.}
    \label{fig:insert1}
  \end{subfigure}
  \begin{subfigure}[b]{0.47\textwidth}
    \hspace{-1cm}
      \begin{tikzpicture}[->,>=stealth',shorten >=1pt,auto,scale = 1,every node/.style={scale=0.65,minimum size=2.5cm,inner sep=.3mm},initial text={}]
        \node[state,initial] (q0) at (0,0) {$q_0$};

        \node[state] (q1) at (2,2) {$q_{c_1}$};
        \node[state] (q2) at (2,0.5) {$q_{c_2}$};
        \node[state] (qn) at (2,-2) {$q_{c_n}$};

        \node (dots) at (2,-0.5) {$ \vdots $};

        \node[state,accepting] (r1) at (4,2) {$r_{c_1}$};
        \node[state,accepting] (r2) at (4,0.5) {$r_{c_2}$};
        \node[state,accepting] (rn) at (4,-2) {$r_{c_n}$};

        \node (dots) at (4,-0.5) {$ \vdots $};

        \path (q0) edge[pos=0.6] node {} (q1);
        \path (q0) edge[pos=0.6] node {} (q2);
        \path (q0) edge[pos=0.35] node {} (qn);

        \node (c1) at (0.8,1.1){$c_1$};
        \node (c2) at (0.8,0.4){$c_2$};
        \node (c3) at (0.8,-1.1){$c_n$};

        \path (q1) edge[pos=0.6] node {} (r1);
        \path (q2) edge[pos=0.6] node {} (r2);
        \path (qn) edge[pos=0.6] node {} (rn);

        \node (d) at (3,2.3){$[d_1,d_n]$};
        \node (d) at (3,0.8){$[d_1,d_n]$};
        \node (d) at (3,-1.7){$[d_1,d_n]$};
      \end{tikzpicture}
      \subcaption{Insert variable greater than $x$ to $ A_\Tau $ in Figure \ref{fig:term1}.}
      \label{fig:insert2}
  \end{subfigure}
}\vspace*{-2mm}
\caption{Insertion of an unused variable with domain $\{d_1,\dots,d_n\}$ into a term automaton}
  \label{fig:insert}\vspace*{-2mm}
\end{figure}
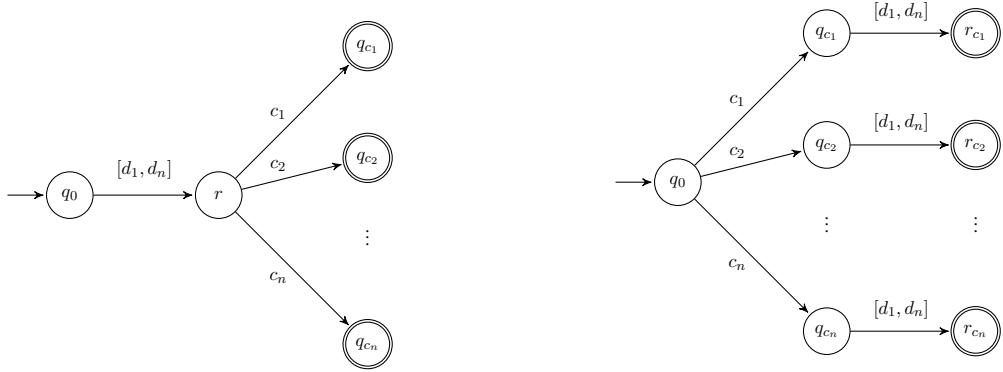

An automaton for an expression is generated using the well-known product automaton construction. However, before applying the product construction, we need to
harmonize the operand automata. Harmonizing means that we need to make sure that they talk about the same set of variables. An unused variable is inserted into a
term automaton or an automaton as shown in Figure \ref{fig:insert}.

Once two automata represent assignments to the same set of variables, we can use the well-known product construction for combining them.
The different operations basically concern the set of final states, so we leave that open for a moment.

\begin{sloppypar}
	\begin{definition}[Product automaton]
	Let $A_1 = [X, Q_1, q_{01}, \delta_1, F_1]$ and $A_2 = [X, Q_2, q_{02}, \delta_2, F_2]$ be automata. Automaton $A = [X, Q, q_0, \delta, F]$ is a product of $A_1$ and $A_2$
	if $Q = Q_1 \times Q_2$, $q_0 = [q_{01},q_{02}]$, and, for all $q_1 \in Q_1$ and $q_2 \in Q_2$ and values $a$, $\delta([q_1,q_2],a) = [\delta_1(q_1,a),\delta_2(q_2,a)]$.
	\end{definition}
\end{sloppypar}

For a comparison $\oplus \in \{=, \leq, \geq\}$, we build the product of two term automata.
The set of final states of the resulting automaton is defined based on the mappings $V_1$ and $V_2$ introduced for term automata: $[q_1,q_2] \in F$ if and only if $q_1 \in F_1$ and $q_2 \in F_2$ and $V(q_1) \oplus V(q_2)$.
For conjunction (resp.~disjunction), the set of final states is defined as follows: $[q_1,q_2] \in F$ if and only if $q_1 \in F_1$ and (resp.~or) $q_2 \in F_2$.
State explosion in the constructions can be alleviated by automata minimization.

 \begin{example}
 	As an example, consider the expression $ x++ = 2 $ and assume that the domain of $ x $ is $ \{1,2,3\} $.
 	The term automata for $ x++ $ and $ 2 $ are depicted in Figure \ref{fig:term}, where the automaton for $ 2 $ needs to be extended to the unused variable $ x $.
 	The result is shown in Figure \ref{fig:extend}. Figure \ref{fig:product} shows the product automaton.
 	Finally, minimization will merge states $[q_1,r_2]$ and $[q_3,r_2]$. From \cite{iddunfolding}, we borrow the idea of merging edges with consecutive annotations into one edge that is annotated with an interval.
 	The final result is shown in Figure \ref{fig:mini}.
 \end{example}

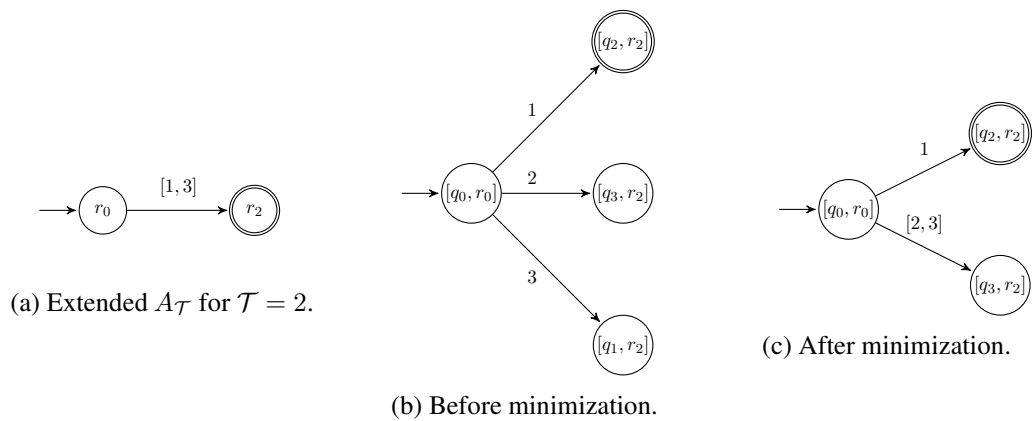
\begin{figure}[ht]
\vspace*{-3mm}
  \begin{subfigure}{0.3\textwidth}
    \hspace{-1cm}
    \begin{tikzpicture}[->,>=stealth',shorten >=1pt,auto,scale = 1,every node/.style={scale=0.65,minimum size=2.5cm,inner sep=.3mm},initial text={}]
        \node[state,initial] (q0) at (0,0) {$r_0$};
        \node[state,accepting] (q1) at (2,0) {$r_2$};
        \path (q0) edge[pos=0.6] node {} (q1);
        \node (c1) at (1,0.3){$[1,3]$};
    \end{tikzpicture}
    \subcaption{Extended $ A_\Tau $ for $ \Tau = 2 $.}
    \label{fig:extend}
  \end{subfigure}
  \begin{subfigure}{0.3\textwidth}
    \hspace{-1cm}
    \begin{tikzpicture}[->,>=stealth',shorten >=1pt,auto,scale = 1,every node/.style={scale=0.65,minimum size=2.5cm,inner sep=.3mm},initial text={}]
      \node[state,initial] (q0) at (0,0) {$[q_0,r_0]$};
      \node[state,accepting] (r1) at (2,2) {$[q_2,r_2]$};
      \node[state] (r2) at (2,0) {$[q_3,r_2]$};
      \node[state] (r3) at (2,-2) {$[q_1,r_2]$};
      \path (q0) edge[pos=0.6] node {} (r1);
      \path (q0) edge[pos=0.6] node {} (r2);
      \path (q0) edge[pos=0.35] node {} (r3);
      \node (c1) at (0.8,1.1){$1$};
      \node (c2) at (0.8,0.2){$2$};
      \node (c3) at (0.8,-1.1){$3$};
    \end{tikzpicture}
    \subcaption{Before minimization.}
    \label{fig:product}
  \end{subfigure}
  \begin{subfigure}{0.3\textwidth}
    \hspace{-0.8cm}
    \begin{tikzpicture}[->,>=stealth',shorten >=1pt,auto,scale = 1,every node/.style={scale=0.65,minimum size=2.5cm,inner sep=.3mm},initial text={}]
        \node[state,initial] (q0) at (0,0) {$[q_0,r_0]$};

        \node[state,accepting] (r1) at (2,1) {$[q_2,r_2]$};
        \node[state] (r2) at (2,-1) {$[q_3,r_2]$};

        \path (q0) edge[pos=0.6] node {} (r1);
        \path (q0) edge[pos=0.3] node {} (r2);

        \node (c1) at (1,0.8){$1$};
        \node (c2) at (1,-0.2){$[2,3]$};
      \end{tikzpicture}
      \subcaption{After minimization.}
      \label{fig:mini}
  \end{subfigure}\vspace*{-2mm}
\caption{Product automaton for expression $x++=2$}\vspace*{-2mm}
\end{figure}

\subsection{Checking full transition classes}

For checking whether a transition class is full, we need to transform assignments that satisfy the guard into token distributions that are consumed from pre-places.
Thanks to our initial simplifications, all we need to do is to project the assignments to the variables that occur on incoming arcs. This can be easily done with the
following considerations.

First, we make sure that the variables at incoming arcs occur first in the order of variables. Second, we make sure that variables of different transitions but concerning the same place
occur in the same position in the respective order.

Next, assume that there are $n$ variables at incoming arcs to a transition and make sure that there is no state in the corresponding automaton that is reached by a sequence of length $n$
and another sequence of different length. If such state exists it can be split into two equivalent states to satisfy the condition. Then, every state $q$ reached by a sequence of length $n$
is made final if a final state is reachable from $q$. This way, all other variables are existentially quantified. The result is the set of all tokens distributions on pre-places that can be
consumed by any enabled firing mode of a transition.

Finally, the resulting automata are combined using the or-product construction. This way, we obtain an automaton that represents all tokens distributions that can be consumed by any transition in the class. The full transition class criterion is satisfied if and only if the resulting automaton accepts every sequence of length $n$. This can be easily seen in the structure of the automaton
if the resulting automaton has been minimized.
\begin{example}
	As an example, consider the transition class $ [t] $ shown in Figure \ref{fig:transitionclass}.
	Assume that the domain of $x$ is $\{1,2,3,4\}$ while the domain of $y$ is $\{1,2,3\}$.
	Let $x< y$.
	Figures \ref{fig:firingmodes1} and \ref{fig:firingmodes2} show the automata representing the token distributions for $t_1$ and $t_2$, respectively.
	The or-product of these automata is shown in Figure \ref{fig:beforemin}.
	Minimization leads to the automaton in Figure \ref{fig:aftermin} from which it is easy to see that the transition class
	is full.
\end{example}

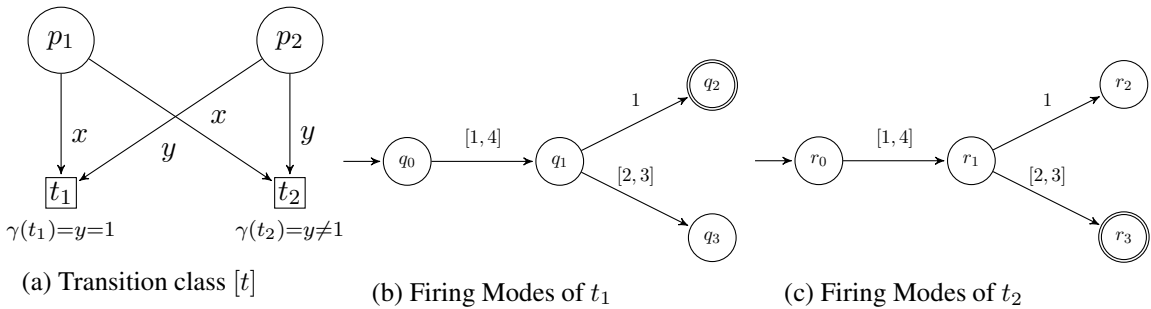
\begin{figure}[ht]
\vspace*{-4mm}
 \resizebox{0.95\textwidth}{!}{
  \begin{subfigure}{0.24\textwidth}
    \begin{center}
    \begin{tikzpicture}[->,>=stealth',shorten >=1pt,auto,scale = 1]
        \node[place]   at(0,0) (p1) {$p_1$};

        \node[place]   at(3,0) (p2) {$p_2$};

        \node[transition] at (0,-2) (t1) {$ t_1$};
        \node (g1) at (0,-2.5){$\scriptstyle \gamma(t_1)= y = 1$};

        \node[transition] at (3,-2) (t2) {$ t_2$};
        \node (g1) at (3,-2.5){$\scriptstyle \gamma(t_2)= y \neq 1$};

        \path (p2) edge[pos=0.6] node {$y$} (t1);
        \path (p2) edge[pos=0.6] node {$y$} (t2);
        \path (p1) edge[pos=0.6] node {$x$} (t2);
        \path (p1) edge[pos=0.6] node {$x$} (t1);
    \end{tikzpicture}
    \caption{Transition class $ [t] $}
    \label{fig:transitionclass}
    \end{center}
  \end{subfigure}
  \begin{subfigure}{0.34\textwidth}
      \hspace{-1cm}
      \begin{tikzpicture}[->,>=stealth',shorten >=1pt,auto,scale = 1,every node/.style={scale=0.65,minimum size=2.5cm,inner sep=.3mm},initial text={}]
        \node[state,initial] (q0) at (0,0) {$q_0$};
        \node[state,accepting] (q1) at (4,1) {$q_2$};
        \node[state] (q2) at (4,-1) {$q_3$};

        \node[state] (r) at (2,0) {$q_1$};

        \path (q0) edge[pos=0.6] node {} (r);
        \node (r3) at (3,0.8){$1$};

        \path (r) edge[pos=0.6] node {} (q1);
        \node (r1) at (1,0.3){$[1,4]$};

        \path (r) edge[pos=0.3] node {} (q2);
        \node (c2) at (3,-0.2){$[2,3]$};
      \end{tikzpicture}
      \subcaption{ Firing Modes of $t_1$}
      \label{fig:firingmodes1}
    \end{subfigure}
    \begin{subfigure}{0.34\textwidth}
    \hspace{-1cm}
      \begin{tikzpicture}[->,>=stealth',shorten >=1pt,auto,scale = 1,every node/.style={scale=0.65,minimum size=2.5cm,inner sep=.3mm},initial text={}]
        \node[state,initial] (q0) at (0,0) {$r_0$};
        \node[state] (q1) at (4,1) {$r_2$};
        \node[state,accepting] (q2) at (4,-1) {$r_3$};

        \node[state] (r) at (2,0) {$r_1$};

          \path (q0) edge[pos=0.6] node {} (r);

        \path (r) edge[pos=0.6] node {} (q1);
        \path (r) edge[pos=0.6] node {} (q2);

        \node (r3) at (3,0.8){$1$};
        \node (r1) at (1,0.3){$[1,4]$};
        \node (c2) at (3,-0.2){$[2,3]$};
      \end{tikzpicture}
      \subcaption{ Firing Modes of $t_2$}
      \label{fig:firingmodes2}
    \end{subfigure}
  }\vspace*{-2mm}
\caption{A Transition class and the automata representing its firing modes}
\end{figure}

\begin{figure}[ht]
\vspace*{-10mm}
 	\resizebox{\textwidth}{!}{
	\begin{subfigure}{0.5\textwidth}
	\begin{center}
		\begin{tikzpicture}[->,>=stealth',shorten >=1pt,auto,scale = 1,every node/.style={scale=0.65,minimum size=2.5cm,inner sep=.3mm},initial text={}]
		  \node[state,initial] (q0) at (0,0) {$[q_0,r_0]$};
		  \node[state,accepting] (q1) at (4,1) {$[q_2,r_2]$};
		  \node[state,accepting] (q2) at (4,-1) {$[q_3,r_3]$};
		
		  \node[state] (r) at (2,0) {$[q_1,r_1]$};

		  \path (q0) edge[pos=0.6] node {} (r);
		  \path (r) edge[pos=0.6] node {} (q1);
		  \path (r) edge[pos=0.3] node {} (q2);
		  \node (r1) at (3,0.8){$1$};
		  \node (r2) at (3,-0.2){$[2,3]$};
		  \node (r1) at (1,0.3){$[1,4]$};
		
		\end{tikzpicture}
		\subcaption{Or-product for $ [t] $ before minimization}
		\label{fig:beforemin}
    \end{center}
	\end{subfigure}
	\begin{subfigure}{0.5\textwidth}
	\begin{center}
		\begin{tikzpicture}[->,>=stealth',shorten >=1pt,auto,scale = 1,every node/.style={scale=0.65,minimum size=2.5cm,inner sep=.3mm},initial text={}]
  			\node[state,initial] (q0) at (0,0) {$[q_0,r_0]$};
  			\node[state,accepting] (q1) at (4,0) {$[q_2,r_2]$};

  			\node[state] (r) at (2,0) {$[q_1,r_1]$};

    		\path (q0) edge[pos=0.6] node {} (r);
    		\node (r1) at (1,0.3){$[1,4]$};
  			\path (r) edge[pos=0.6] node {} (q1);
			\node (r1) at (3,0.3){$[1,3]$};
		\end{tikzpicture}
		\subcaption{Or-product for $ [t] $ after minimization}
		\label{fig:aftermin}
  	\end{center}
    \end{subfigure}
	}\vspace*{-2mm}
\caption{Checking for full transition class}
\end{figure}
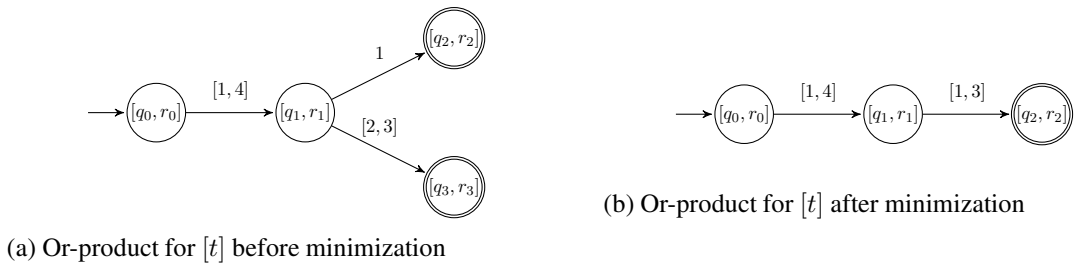

Our experiments revealed that we are able to decide the full transition class criterion for coloured nets where, so far, no participant in
the model checking contest could create its unfolding. We managed to get some verification results for those nets using the skeleton approach.

\section{Extension to place/transition nets} \label{sec:ext}

With the results presented so far, the skeleton abstraction is only available for systems modeled as coloured Petri nets.
In this section, we extend the applicability to nets that are originally modeled as P/T nets.
There exist translations from various high level system descriptions directly into P/T nets that could as well have been translated into coloured nets.
We propose an efficient procedure to fold a P/T net $ N $ into a coloured net $ C_{\N} $, for which we then can build the skeleton $ S_{\N} $.

\subsection{Folding P/T nets} \label{subsec:fold}

The idea of folding a P/T net into a coloured net is as old as coloured nets as such.
To our best knowledge, however, the efficiency of an actual implementation has not been observed so far.
Our approach is based on partition refinement.
The goal here is to partition a set $ \cal{M} $ into a partition $ M $ of disjunct subsets $ M_1,\dots,M_n $.
First, $ M $ contains only one subset, which is $ \cal{M}$.
The partition is then refined by the application of a split function.

\begin{definition}
	Let $ M = \{M_1, \dots, M_n\} $ be a partition of the set $ \cal{M} $ and $ f: \cal{M} \rightarrow \mathbb{Z} $ be a split function.
	The application of $ f $ on the partition $ M $ is defined as:
	$ \operatorname{split}(M,f)=\{\{x \mid x \in M_i, f(x) = j\}\} $
	for $ 1 \leq i \leq n, j \in \mathbb{Z}$.
\end{definition}

Informally, we separate elements, where $f$ yields different values.
This leads to a new partition of $ \cal{M} $.
For two subsets $ M_i,M_j $, it should hold that $ M_i \not = \emptyset $, $ M_i \cap M_j = \emptyset $ and also $ M_1 \cup \dots \cup M_n = \cal{M} $.
For implementing a split operation, we assume an array where every element of $ \cal{M} $ appears exactly once.
For every class in the partition, there is a pair of indices $ i $ and $ j $ such that the elements of the class are the array entries between $ i $ and $ j $.
For a split operation, we separately sort the elements of each class and then introduce new classes where adjacent elements have different $ f $-values.
Given a P/T net $ N = [P,T,F,W,m_0] $, the initial set, which should be partitioned is $ \cal{M} $ $= P \cup T $.
The coarsest partition  fitting all requirements is $M = \{P,T\}$.
We refine this partition such that, ultimately, every class of places of the given net serves as a place of the resulting coloured net $ C_{\N} $ while every class of transitions of the given net serves as a transition.

While folding, We need to conform the restrictions of uniformity, that building a skeleton is possible.
The uniformity criterion used here is more liberal than the one in \cite{ourskeleton}. There, we required that, for every two low level transitions $t_1$ and $t_2$ in some transition class,
every place class $M$
and every weight $k$, $t_1$ and $t_2$ have the same number of pre-places (resp.~post-places) with weight $k$ in $M$. Here, we only require that the sum of all weights between $t_1$
and places in $M$ is the same as for $M_2$.

This modification yields coarser classes (thus smaller skeletons) and the algorithms run faster.
This way, in an otherwise equivalent experimental setting, the number of cases where the skeleton approach responded as the fastest member of our portfolio rose from 3168 to 3702.
The number of queries we could answer but no participant of the model checking contest could answer in 2019 climbed from 226 to 248.

The folding happens with regard to an $ ACTL^* $ formula $ \phi $.
Let $ AP_{\phi} $ denote the set of atomic propositions occuring in $ \phi. $
The procedure for folding a P/T net into a coloured net is described in Figure~\ref{alg:folding}.

\begin{figure}[t]
	\begin{algorithmic}
		\State \textbf{Input:} Petri net $ N = [P,T,F,W,m_0] $
		\State \textbf{Output:} Partition $ M $ of $ P \cup T $ resp. $ C_{\N} = [P_{\C\N},T_{\C\N},F_{\C\N},W_{\C\N},\chi,\gamma,m_{0\C\N}]$
		\State Let $ M = \{P,T\} $;
		\State $ M = split(M,f) $ where $ f(x) = \operatorname{card}(\pre{x}) $;
		\State $ M = split(M,f) $ where $ f(x) = \operatorname{card}(\post{x}) $;
		\ForAll{atomic propositions $p \in AP_{\phi}$ with the form $ k_1 p_1 + \dots + k_n p_n <= k $}
			\State $ M = split(M,f) $ where for all occuring places $ p_i \in p : f(p_i) = k_i $ for $ \scopen{i} $, else $ f(x) = 0 $;
			\EndFor
		\ForAll{place classes $ M^* \in M $}
				\State $ M = split(M,f) $ where $ f(x) =   \sum_{p \in M^*} W(p,x)$;
				\State $ M = split(M,f) $ where $ f(x) =   \sum_{p \in M^*} W(x,p)$;
		\EndFor
		\State $P_{\C\N}$ = place classes of $ M $, $T_{\C\N}$ = transition classes of $ M $;
		\State $ (M^*,M^{*\prime}) \in F_{\C\N}$, $ W_{\C\N}(M^{*},M^{*\prime}) = \{x_1,\dots,x_k\}$, iff $ \exists x \in M^{*}, \exists y \in M^{*\prime} :$
		\State $ (x,y) \in F, W(x,y) = k$ for $ k \in \mathbb{N} $ and $ M^{*},M^{*\prime} \in M$;
		\State $ m_{0\C\N}(M^{*}) = \sum_{p \in M^{*}} m_{0}(p)$,
		$ \chi(M^{*}) = \{p \mid p \in M^* \} $ for every place class $ M^{*} \in M$;
		\State $ \gamma(M^{*}) = \bigvee_{t_i \in M^*} \big( \big(\bigwedge_{M^{*\prime} = \{p_1,\dots,p_n\} \in \pre{M^*}} \bigwedge_{j=1}^n \bigwedge_{k=1}^{W(p_j,t_i)} x_{p_{\sum_{\ell = 1}^{j-1}W(p_\ell,t_i)+k}} = p_j \big)\wedge \big(\bigwedge_{M^{*\prime} = \{p_1,\dots,p_n\} \in \post{M^*}} \bigwedge_{j=1}^n \bigwedge_{k=1}^{W(t_i,p_j)} x_{p_{\sum_{\ell = 1}^{j-1}W(p_\ell,t_i)+k}} = p_j\big)\big)$
		 for every transition class $ M^{*} \in M $;
	\end{algorithmic}\vspace*{-1mm}
	\caption{Algorithm for folding a P/T net into a coloured net.}
	\label{alg:folding}\vspace*{-2mm}
\end{figure}

In this procedure, the guard expressions $ \gamma $ deserve additional explanation.
If $ M^* $ is a transition in the folded net, its elements $ \{t_1,\dots,t_m\} $ serve as firing modes of transition $ M^*$ .
For each of these firing modes, we need to specify the effect on the pre-places and post-places.
That is why the general structure of the guard starts with $ \bigvee_{t_i \in M^*} \dots $ \,.
Since every pre- and every post-place needs to be considered, we have the conjunctions $ \bigwedge_{M^{*\prime} = \{p_1,\dots,p_n\} \in \pre{M^*}} \, \dots $ and $ \bigwedge_{M^{*\prime} = \{p_1,\dots,p_n\} \in \post{M^*}} \, \dots $ \,.
The remaining content of the guard expressions is the same for pre- and post-places.
It specifies for a given firing mode (i.e.~low level transition $ t_i $) how many tokens of some colour (i.e.~low level place $ p_j \in {M^*}' $) need to be consumed or produced.
This number is $ W(p_j,t_i) $ (or $ W(t_i,p_j) $, respectively).
For consuming $ k $ tokens of some colour $ p_j $ (leading to conjunction $ \bigwedge_{j=1}^n\dots $), we need to bind $ k $ of the arc variables to colour $ p_j $ (leading to conjunction ($ \bigwedge_{k=1}^{W(p_j,t_i)}\dots $).
If these bindings are done in consecutive order of the low level places $ p_1, \dots, p_n $, the $ k $ variables bound to $ p_j $ are those that start with index $ \sum_{\ell = 1}^{j-1}W(p_\ell,t_i) $ (or $ \sum_{\ell = 1}^{j-1}W(t_i,p_\ell) $).

To make this algorithm more understandable, we demonstrate it with an example.

\begin{example}
	\label{ex:folding}
	We consider a P/T net $ N $, which shows the dilemma of five dining philosophers.
	The P/T net is structured as follows:
	For every $i \in \{0,\dots,4\}$ there is a place $ th_i $ (philosopher $ i $ is thinking), a place $ hl_i $ (has left fork),
	$ hr_i $ (has right fork), $ ea_i $ (philosopher $ i $ is eating) and $ fo_i $ (fork $ i $ is on the table).
	There are the transitions $ tl_i $ (take left fork) that consume tokens from $ th_i $ and $ fo_i $, and produce on $ hl_i $,
	transitions $ tr_i $ (take right fork) that consume from $ hl_i $ and $ fo_{i+1 \operatorname{mod} 5} $ and produce on $ ea_i $, transitions $ rl_i $ (release left fork) that consume from $ ea_i $ and produce on $ hr_i$ and $fo_i$, and, finally, transitions $ rr_i $ (release right fork) that consume from $ hr_i $ and produce on $ fo_{i+1 \operatorname{mod} 5} $ and $ th_i $.
	Places $ th_i $ and $ fo_i $ are initially marked, and all arc weights are 1.
	For better readability, $ x_i $ describes the set $ x_0,\dots,x_4 $ for every node $ x \in P \cup T$ of the net.
	The folding is regarding the $ ACTL^* $ formula $\phi:  \mathbf{A}\mathbf{G} \,\, \neg ( \sum_i hr_i = 4)  \,\, \land \,\,  (\sum_i hl_i = 1) $ for $ i \in \{0,\dots,4\} $.
	Initially the coarsest partition distinguishes between places and transitions: $ M = \{\{th_i,ea_i,fo_i,hl_i,hr_i\},\{tr_i,tl_i,rl_i,rr_i\}\}$.
	Then, the sets are split according to the \emph{number of incoming} and \emph{outgoing arcs}.
	All places $ fo_i $ have two incoming and two outgoing transitions, all remaining places only have one incoming and one outgoing transition.
	The $ tr_i $ and $ tl_i $ transitions have two incoming places and one outgoing place, $ rl_i $ and $ rr_i $ the other way round.
	This leads to partition $ M = \{\{th_i,ea_i,hl_i,hr_i\},\{fo_i\},\{tr_i,tl_i\},\{rl_i,rr_i\}\}$.
	The \emph{atomic propositions} of $ \phi $ give additional restrictions, as the elements of the subsets finally should 	satisfy those propositions equally.
	So, for the atomic propositions $ \sum_{i=0}^4 hr_i = 4  $ and $ \sum_{i=0}^4 hl_i = 1 $, every place is mapped to its 	coefficient in the corresponding proposition.
	Transitions are not affected here, so $ M = \{\{th_i,ea_i\},\{hl_i\},\{hr_i\},\{fo_i\},\{tr_i,tl_i\},\{rl_i,rr_i\}\}$.
	For obtaining \emph{uniformity}, we split $\{tr_i, tl_i\}$ into $\{tr_i\}$ and $\{tl_i\}$ since every transition $tl_i$ has a pre-places in $\{th_i,ea_i\}$ while no transition $tr_i$ has,
	and we split $\{rl_i,rr_i\}$ into $\{rl_i\}$ and $\{rr_i\}$ since every transition $rl_i$ has a pre-place in $\{th_i, ea_i\}$ while no transition $rr_i$ has.	
 	The resulting partition is $ M = \{\{th_i,ea_i\},\{hl_i\},\{hr_i\},\{fo_i\},\{tr_i\},\{tl_i\}, $ \\
 	$ \{rl_i\},\{rr_i\}\} $ and is uniform.

\medskip	
	Finally, every place class $ M_i $ is turned into a place $ p_{M_i} $ with $ \sum_{p \in M_i} m(p) $ tokens and the colour domain $ \chi(p_{M_i}) = M_i $, and every transition class $ M_j$ into a transition $t_{M_j}$.
	We obtain the places $ thea,fo,hl,hr $ and the transitions $ tr,tl,rl,rr $.
	Let there be an arc from $ p_{M_i} $ to $t_{M_j}$, if there exists some $ p \in M_i  $ and $ t \in M_j $ with $ (p,t) \in F$.
	Arcs from transitions to places are formed analogously.
	An Arc $(p_{M_i},t_{M_j})$ ist assigned with the variables $ x_1,\dots,x_{w} $ where $ w = W(p,t)$ for $ p \in M_i$ and $ t \in M_j $.
	In the end, we need to formulate the guard of the transitions, which needs to ensure, that the coloured transition only fires if the right coloured tokens lay on the pre-places.
	We thefore build the disjunction of the input requirements of the transitions according the arcs and weights of $ N $.
	As an example, the guard of transition $ tl $ is
	\begin{align*}
		\gamma(tl) &= x_{11} = th_0 \land x_{12} = fo_0 \land x_1 = hl_0 \hspace{0.7cm} \text{(firing mode $ tl_0 $)} \\
				   &\lor x_{11} = th_1 \land x_{12} = fo_1 \land x_1 = hl_1 \hspace{0.7cm} \text{(firing mode $ tl_1 $)} \\
				   &\lor x_{11} = th_2 \land x_{12} = fo_2 \land x_1 = hl_2 \hspace{0.7cm} \text{(firing mode $ tl_2 $)} \\
				   &\lor x_{11} = th_3 \land x_{12} = fo_3 \land x_1 = hl_3 \hspace{0.7cm} \text{(firing mode $ tl_3 $)} \\
				   &\lor x_{11} = th_4 \land x_{12} = fo_4 \land x_1 = hl_4 \hspace{0.7cm} \text{(firing mode $ tl_4 $)}
	\end{align*}

	Figure \ref{fig:philos} presents the coloured net, which results from the described folding.
	The coloured net can subsequently be decoloured to a skeleton.

	\begin{figure}[ht]
	\begin{center}
	\resizebox{0.7\textwidth}{!}{
	\begin{tikzpicture}
		[scale=0.5,
		 place/.style={circle,draw=black,thick,inner sep=10pt},
		 transition/.style={rectangle,draw=black,thick,inner sep=10pt},
		 token/.style={shape=circle,draw=black,fill=white,minimum size=2mm,inner sep=0.2pt},
		 -|/.style={to path={-| (\tikztotarget)}},
		 |-/.style={to path={|- (\tikztotarget)}}
		]
		\node[place,label=below:$ fo $] at(16,-4) (fo) {};
			\node[token] at(16,-3.35) {\scalebox{.47}{fo0}};
			\node[token] at(16.65,-4) {\scalebox{.47}{fo1}};
			\node[token] at(16,-4.65) {\scalebox{.47}{fo2}};
			\node[token] at(15.35,-4) {\scalebox{.47}{fo3}};
			\node[token] at(16,-4) {\scalebox{.47}{fo4}};

		\node[place,label=above:$ thea $] at(16,4) (thea) {};
			\node[token] at(16,4.65) {\scalebox{.47}{th0}};
			\node[token] at(16.65,4) {\scalebox{.47}{th1}};
			\node[token] at(16,3.35) {\scalebox{.47}{th2}};
			\node[token] at(15.35,4) {\scalebox{.47}{th3}};
			\node[token] at(16,4) {\scalebox{.47}{th4}};
			
		\node[place,label=above:$ hl $] at(8,0) (hl) {};
		\node[place,label=above:$ hr $] at(24,0) (hr) {};

		\node[transition] at (4,0) (tl) {$ tl $}
			edge[pre,|-] (thea)
			edge[pre,|-] (fo)
			edge[post] node[auto] {\footnotesize $ x_1 $} (hl);

		\node[transition] at (12,0) (tr) {$ tr $}
			edge[pre] node[auto,swap] {\footnotesize $ x_2 $} (hl)
			edge[pre,|-] (fo.150)
			edge[post,|-] (thea.210)
			;
		
		\node[transition] at (20,0) (rl) {$ rl $}
			edge[pre,|-] (thea.330)
			edge[post] node[auto] {\footnotesize $ x_5 $} (hr)
			edge[post,|-](fo.30)
			;

		\node[transition] at (28,0) (rr) {$ rr $}
			edge[pre] node[auto,swap] {\footnotesize $ x_7 $} (hr)
			edge[post,|-](thea)
			edge[post,|-](fo.330)
			;

		\node at(8,4.5) (xthtl) {\footnotesize $x_{11} $};

		\node at(11.5,2.3) (xthtl) {\footnotesize $x_3 $};
		\node at(11.5,-2.3) (xthtl) {\footnotesize $x_8 $};

		\node at(20.5,2.3) (xthtl) {\footnotesize $x_4 $};
		\node at(20.5,-2.3) (xrlfo) {\footnotesize $x_6 $};

		\node at(24,4.5) (xrrth) {\footnotesize $x_9 $};
		\node at(24,-4) (xrrfo) {\footnotesize $x_{10} $};

		\node at(8,-3.5) (xfotl) {\footnotesize $x_{12} $};
		
	\end{tikzpicture}}
	\end{center}\vspace*{-3mm}
	\caption{A coloured net version of the five dining philosophers.}
	\label{fig:philos}
\end{figure}
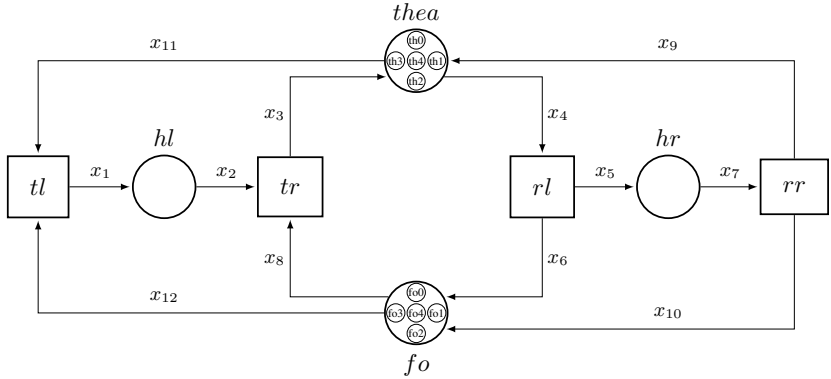\vspace*{-1mm}
\end{example}

\subsection{Checking fullness for a P/T net}

The resulting coloured net does not necessarily have full minimal transition classes (cf.~Sec.~~\ref{subsec:findlow}), thus it does not have a deadlock-preserving skeleton.
Due to the process for deriving the folded coloured net, the guards do not permit the approach outlined there for checking whether or not a transition class is full.
We may, however, approach that criterion differently.
The idea is to check the criterion right after the folding procedure, just as we have the final partitioning of the P/T nodes.
If we cannot prove the deadlock preservation at this point, we can abort the skeletal analysis of this net, as we don't expect useful results.

\begin{theorem}(Deadlock Preservation for P/T Nets)
	Let $ N $ be a P/T net and $ C_{\N} $ its folding.
	Let $ [t] = \{t_1,t_2,\dots,t_k\} $ be a minimal transition class of $ C_{\N} $, where each transition $ t_j $ has $ s_j $ firing modes for $ j \in \{1,\dots,k\} $.
	Let $ p_1,p_2,\dots,p_\ell $ be the pre-places of $ [t] $, with the colour domains $ \chi(p_i) $ for $ i \in \{1,\dots,\ell\} $.
	Every pre-place $ p_i $ is connected to every transition $ t_j $ of $ [t] $ by an arc with the weight $ w_{ij} $.
	The folding $ C_{\N} $ resp. the underlying P/T net $ N $ has a deadlock-preserving skeleton, if $ \prod_{i=1}^\ell \binom{\abs{\chi(p_i)}}{w_{ij}}  = \sum_{j=1}^k s_j $ for every minimal transition class of $ C_{\N} $.
\end{theorem}

\begin{proof}
	The folding $ C_{\N} $ has a deadlock-preserving skeleton, if all of its minimal transition classes are full.
	A minimal transition class $ [t] = \{t_1,t_2,\dots,t_k\}$ is full, if for every marking $ m_{C_{\N}} $ with $ \abs{m_{C_{\N}}(p_i)} = f_i(t) $ for $ i \in \{1,\dots,\ell\} $, there is one transition in $ [t] $, for which the marking is a firing mode.
	This is expressed by the equation $ \prod_{i=1}^\ell \binom{\abs{\chi(p_i)}}{w_{ij}}  = \sum_{j=1}^k s_j $.
	The binomial coefficient $ \binom{\abs{\chi(p_i)}}{w_{ij}} $ gives the number of sufficient tokensubsets of $ \chi(p_i) $ for one pre-place $ p_i $ of $ [t] $, where $ i \in \{1,\dots,\ell\} $ and $ j \in \{1,\dots,k\} $.
	Multiplying these numbers for every pre-place $ p_i $ for $ i \in \{1,\dots,\ell\} $, leads to the total number of sufficient combinations of tokens, i.e the number of possible markings $ m_{C_{\N}} $ with sufficient input requirements $ \abs{m(p_i)} = f_i(t) $ with $ i \in \{1,\dots,\ell\} $ for $ [t] $.
	Each of the transitions $ t_j $ in $ [t] $ for $ j \in \{1,\dots,k\} $ has $ s_j $ firing modes, thus in $ [t] $, we have $ \sum_{j=1}^k s_j $ firing modes overall.
	If $ \prod_{i=1}^\ell \binom{\abs{\chi(p_i)}}{w_{ij}}  = \sum_{j=1}^k s_j $, for every sufficient combination of the coloured tokens, there is a firing mode of one transition in $ [t]$.
	If $ \prod_{i=1}^\ell \binom{\abs{\chi(p_i)}}{w_{ij}}  = \sum_{j=1}^k s_j $, the minimal transition class $ [t] $ is full, thus if the equation holds for every minimal transition class, $ C_{\N} $ has a deadlock-preseving skeleton.
	Transferring this to $ N $, for every combination of tokens of the P/T pre-places (represented by $ \chi(p_i) $), there is one P/T transition related to $ [t] $ enabled (representend by $ s_j $).
	So, if a marking of $ N $ fits with regard to the cardinality, there must be one activated P/T transition if the equation holds.
	Then, $ N $ is has a deadlock-preserving skeleton.
	It it important to mention, that the firing modes $ s_j $ need to be all different from each other, resp. all of the P/T transitions need to have different presets.
	Otherwise the equality of combinations and firing modes will not hold, although every combination activates a transition.
\end{proof}

This equation is sufficient for the fullness of $ [t] $, but it is not necessary.
\newline If $ \prod_{i=1}^l \binom{\abs{\chi(p_i)}}{w_{ij}}  > \sum_{j=1}^k s_j $, which means there is a combination of tokens which does not activate a transition, these too many combinations might be unreachable, thus are not in need of an activated transition.
If the equation holds for every minimal transition class we know that $ C_{\N} $ will have a deadlock-preserving skeleton and the method of skeletal abstraction can be applied to $ N $ and all its $ ACTL^* $ formulas.

\section{Experimental results}
\label{sec:experiments}

We conducted our experiments on the benchmark provided by the Model Checking Contest (MCC) 2019 \cite{mcc19}. On that page, the reader may find a detailed specification
of the machine ``tajo'' that was used to execute the experiments. The benchmark comprises 1018 nets (193 coloured nets and 825 P/T nets). For the majority of colored nets, their
unfolding is among the P/T nets of the benchmark, too.
We covered the three categories Reachability, CTL, and LTL where the skeleton approach makes sense. For every net and category, there are 16 formulas with place-based atomic propositions
and 16 formulas with transition-based propositions. That makes a total of 97,728 formulas. If a P/T net is the unfolding of a coloured net, some but not all formulas
of the P/T net accord with formulas used for the coloured net.

In every single run, we allowed 4 cores, 30 minutes, and 16 MB of RAM for the verification of a group of 16 formulas.
The runs used the full portfolio \cite{portfolio} of verification methods available in our tool LoLA \cite{lola}, now including the skeleton approach. For the skeleton, we applied the same search based model checking routines as for the unfolded net, and the state equation approach \cite{wimmelwolf}.
In our approach, the skeleton is directly derived from the PNML description of a coloured net, so the
skeleton related verification tasks start before the unfolding of the net is generated  in parallel (and only then the remaining verification routines are launched). If the input is a P/T net, we first launch the
verification tasks for the given net, before trying to fold the net (in parallel to the already running routines). We launch skeleton related tasks only if the size of the skeleton is less than one third of the size of the given P/T net. This way, we avoid situations where the skeleton is too close to the given net. Since folding depends on the formula, we
have to execute up to 16 individual folding procedures per run. We do not fold a net if the formula is trivial (i.e.~does not contain temporal operators). Trivial formulas are
mostly the result of sophisticated application of logical tautologies and preprocessing based on linear programming \cite{tapaalrewriting}. We also stop the folding procedure as soon as some
other portfolio member has determined the value of the formula.

For the 79,200 formulas for P/T nets, 66,546 skeletons were created. Of the remaining 12,654 formulas, 11,086 contain no temporal operators, so no folding was launched.
For the remaining 1,568 formulas, some
other portfolio member may have delivered a result before folding completed.  Folding took at most 287 seconds, with an average of half a second. Generation of the skeleton for a coloured net takes
no time at all as it appears as an intermediate step of the unfolding process. Generating the skeleton naturally succeeded for all coloured nets. It also succeeded whenever both net and formula were derived from a
coloured net. There are a few other cases where the skeleton could be generated. Although more nets have a regular, foldable structure, formulas, if not derived from coloured nets, are
generated randomly, so they tend to break symmetry more frequently than in practical situations. On the other hand, the formula syntax for coloured nets in the MCC does not permit
references to individual colours  or firing modes, so the skeleton approach is applicable more frequently than in practice.  Since there are more P/T nets than coloured nets in the contest,
results obtained for the MCC benchmark should be a lower bound for the performance to be observed in practice.

The 66,546 skeletons include those that are considered to be too large to make a difference compared to the given net. After ruling them out, 34,906 formulas have useful skeletons, including
coloured and P/T nets. In 15,315 cases, we launched the skeleton related tasks. In the remaining 19,591 cases, the formula (nor its negation) are not in $ ACTL^*$, or none of the criteria
discussed in the paper would certify preservation of the formula. We need to mention here that deadlock injection has not been implemented so far.

Of the 15,315 formulas where we launched the skeleton related tasks, they were the first (among the whole portfolio) to deliver results in 3702 cases.
The remaining 12,147 formulas include those where some other portfolio member responded earlier, or where the skeleton approach evaluated its $ ACTL^* $
query to false (so the value is not inherited by the unfolded net).

Among the 97,728 formulas considered, there have been 7,768 formulas none of the participants in the MCC 2019 could solve. With the skeleton
approach, we have now been able to solve 248 of these particularly involved problems. These include but are not restricted to nets that have a prohibitively large unfolding.

Given that we run the skeleton approach as part of a powerful portfolio,
with LoLA being a competitive participant in the MCC, we may conclude that the skeleton approach nicely complements the existing portfolio.

\section{Conclusion} \label{sec:conclusion}
With our contribution, we turned the skeleton approach into an executable and useful member of a verification portfolio. We investigated the gap between the concepts of
net morphisms and simulation relations and proposed algorithms for checking the required criteria. Through folding, we extended the approach to P/T nets. Experiments underpin
the usefulness of the approach.

Future work may include the implementation of deadlock injection.
Furthermore, we may enhance the approach to the full transition classes.
First, we may try to use place invariants to rule out certain token distributions in the pre-set of transition classes thus being able to certify more transition classes as full.
Second, we may try to split places and transitions in the skeleton turning non-full transition classes into full ones.
Furthermore, we may determine a larger number of $ LTL $ safety properties by the analysis of the respective Büchi automaton.
So far, we assume all $ LTL $ properties which do not use $ \mathbf{X},\mathbf{F} $ and $ \mathbf{U} $ as safety properties.
Libraries like \cite{spot} are able to identify more $ LTL $ properties as safety properties.
This might extend the use of the skeleton approach.

\end{document}